\documentclass[5p,times,twocolumn]{elsarticle}
\usepackage[utf8]{inputenc}
\usepackage{graphicx}
\usepackage{amssymb}
\usepackage{color}
\usepackage{lineno}
\usepackage[ruled,vlined,linesnumbered,titlenumbered]{algorithm2e}
\usepackage{amsmath}
\usepackage{times}
\usepackage{optidef}
\usepackage{amsthm}
\usepackage{multirow}
\usepackage{makecell}
\usepackage{enumitem}
\usepackage{booktabs}
\usepackage[edges]{forest}
\usepackage{pifont}
\usepackage{soul}
\usepackage{amssymb}
\usepackage{url}
\SetKwInput{Parameters}{Parameters}
\usepackage{threeparttable}
\usepackage{adjustbox}
\usepackage{color}
\usepackage{hyperref}
\usepackage[dvipsnames]{xcolor}

\usetikzlibrary{positioning}
\usetikzlibrary{shapes.geometric, arrows}
\tikzstyle{starosaba2025dtstop} = [rectangle, rounded corners, minimum width=3.5cm, minimum height=0.7cm, text centered, draw=black, fill=blue!20]
\tikzstyle{process} = [rectangle, minimum width=3.5cm, minimum height=0.7cm, text centered, draw=black, fill=green!20]
\tikzstyle{arrow} = [thick,->,>=stealth]
\tikzstyle{comment} = [text width=3cm, text centered, font=\small, fill=white]

\usepackage[most]{tcolorbox}
\newtcbtheorem{definition}{Definition}{%
  colback=gray!10,  
  colframe=black,   
  coltitle=gray!10,   
  fonttitle=\bfseries, 
  sharp corners,    
  boxed title style={size=small,colback=gray!10},
}{def}

\newtcbtheorem{proposition}{Proposition}{%
  colback=gray!10,
  colframe=black,
  coltitle=gray!10,
  fonttitle=\bfseries,
  sharp corners,
  boxed title style={size=small,colback=gray!10},
  before upper={\strut},
  breakable
}{prop}

\newtcbtheorem{Summary}{}{%
  enhanced,
  drop shadow={black!50!white},  
  coltitle=black,
  top=0.3in,
  attach boxed title to top left={xshift=1.5em,yshift=-\tcboxedtitleheight/2},
  boxed title style={size=small,colback=cyan},
}{summary}

\newtcbtheorem{Example}{}{%
  enhanced,
  drop shadow={black!50!white},  
  coltitle=black,
  top=0.3in,
  attach boxed title to top left={xshift=1.5em,yshift=-\tcboxedtitleheight/2},
  boxed title style={size=small,colback=pink},
}{example}

\newtcbtheorem{Futurework}{}{%
  enhanced,
  drop shadow={black!50!white},  
  coltitle=black,
  top=0.3in,
  attach boxed title to top left={xshift=1.5em,yshift=-\tcboxedtitleheight/2},
  boxed title style={size=small,colback=lime},
}{summary}

\newtcbtheorem{Futurework2}{}{%
  enhanced,
  drop shadow={black!50!white},  
  coltitle=black,
  top=0.3in,
  attach boxed title to top left={xshift=1.5em,yshift=-\tcboxedtitleheight/2},
  boxed title style={size=small,colback=orange},
}{summary}

\useforestlibrary{linguistics}
\forestapplylibrarydefaults{linguistics}

\journal{Future Generation Computer Systems}

\begin{document}

\begin{frontmatter}

\title{Addressing the Minor-Embedding Problem in Quantum Annealing \\ and Evaluating State-of-the-Art Algorithm Performance}

\author[tec,upv]{Aitor Gómez-Tejedor}
\ead{aitor.gomez-tejedor@tecnalia.com}
\author[tec]{Eneko Osaba}
\ead{eneko.osaba@tecnalia.com}
\author[tec]{Esther Villar-Rodriguez}
\ead{esther.villar@tecnalia.com}

\address[tec]{TECNALIA, Basque Research and Technology Alliance (BRTA), P. Tecnologico, Ed. 700, 48160 Derio, Spain}
\address[upv]{University of the Basque Country (UPV/EHU), 48013 Bilbao, Spain}

\begin{abstract}
This study addresses the minor-embedding problem, which involves mapping the variables of an Ising model onto a quantum annealing processor. The primary motivation stems from the observed performance disparity of quantum annealers when solving problems suited to the processor's architecture versus those with non-hardware-native topologies. Our research has two main objectives: \textit{i)} to analyze the impact of embedding quality on the performance of D-Wave Systems quantum annealers, and \textit{ii)} to evaluate the quality of the embeddings generated by Minorminer, the standard minor-embedding technique in the quantum annealing literature, provided by D-Wave. Regarding the first objective, our experiments reveal a clear correlation between the average chain length of embeddings and the relative errors of the solutions sampled. This underscores the critical influence of embedding quality on quantum annealing performance. For the second objective, we evaluate Minorminer’s embedding capabilities, the quality and robustness of its embeddings, and its execution-time performance \textcolor{black}{on Erdös-Rényi graphs}. We also compare its performance with Clique Embedding, another algorithm developed by D-Wave, which is deterministic and designed to embed fully connected Ising models into quantum annealing processors, serving as a worst-case scenario. The results demonstrate that there is significant room for improvement for Minorminer, suggesting that more effective embedding strategies could lead to meaningful gains in quantum annealing performance.
\end{abstract}


\end{frontmatter}

\section{Introduction}\label{sec:intro}

Quantum computing (QC) is recognized as the next significant advancement in computing, with substantial potential across various fields, including combinatorial optimization. In this domain, QC can address problems in fundamentally different ways compared to classical methods, potentially offering significant advantages. 

Despite its promising potential, we are currently in the \textit{noisy intermediate scale quantum (NISQ)} \cite{preskill2018quantum} era, which means that the technology is still in early development and faces hardware limitations. These limitations include the poor quality of qubits, which causes undesirable behavior, and a reduced number of qubits and couplers (connections between pairs of qubits) in the processors. In spite of these restrictions, extensive research is being conducted on the software and algorithmic aspects of the paradigm. Much of this research aims to address these constraints and maximize the potential of current NISQ-era quantum computers. Examples of research areas within this category include the design of hybrid methods that integrate classical and quantum processors \cite{peruzzo2014variational, farhi2014quantum, HSS}, as well as error correction and error mitigation techniques \cite{lidar2013quantum, cai2023quantum}.

Among the various types of quantum computers, the gate-model and quantum annealing are the most prominent. This study focuses on the latter. Quantum annealing (QA) is grounded in the principles of adiabatic quantum computing, as initially proposed by Farhi et al. \cite{farhi2000quantum}. The adiabatic quantum computing paradigm starts by initializing a quantum system in an easy-to-prepare ground state. Then, the system's time-dependent Hamiltonian is slowly evolved within the adiabatic theorem \cite{messiah1961quantum} regime until its final form is reached, where the state desired to be computed is encoded as the ground state. By evolving the system slowly enough, the adiabatic theorem states that there is a high probability that the final state of the system will be the desired ground state. Quantum annealing modifies this approach by relaxing the requirement for slow evolution, which is dependent on the problem and can be exponentially large. This relaxation allows for faster execution times, albeit with a reduced probability of achieving the ground state.

There are different models for building quantum annealers, such as those based on Rydberg atoms \cite{glaetzle2017coherent}, trapped ions \cite{raventos2018semiclassical} or superconducting flux qubits \cite{johnson2011quantum}, with the latter being the most recognized to date. Additionally, several companies are working on this technology and building their own devices, such as NEC\footnote{\url{https://parityqc.com/a-new-quantum-annealer-by-nec}}, Qilimanjaro\footnote{\url{https://qilimanjaro.tech}}, and D-Wave Systems\footnote{\url{https://www.dwavesys.com}}.

In the literature, the reference quantum annealers are those developed by
D-Wave Systems, which are specifically designed to solve Ising models by sampling low-energy states of a given Hamiltonian in which the Ising model is encoded \cite{dwave2025}. Consequently, these annealers forfeit the computational universality inherent in general quantum annealing. However, the Ising model is mathematically equivalent to solving a quadratic unconstrained binary optimization (QUBO) model, and many combinatorial optimization problems can be formulated as QUBOs \cite{glover2018tutorial}, and therefore, as Ising models \cite{lucas2014ising}.

This study exhaustively analyzes the \textit{minor-embedding} problem, a graph theoretical problem responsible for mapping the variables of the Ising model onto the quantum annealing processor. The motivation lies in the observed performance disparity of D-Wave's processors when addressing hardware-native problems versus real-world optimization problems with generic topologies. Specifically, quantum annealers exhibit competitive performance with certain classical solvers when solving problems that are inherently suited to the processor's architecture, thereby eliminating the need for minor-embedding, as shown in \cite{tasseff2024emerging} and \cite{bauza2024scaling}. However, when applied to optimization tasks whose Ising model formulation has a non-hardware-native topology, the solution quality produced by quantum annealers is significantly inferior to that achieved by classical solvers \cite{mohseni2022ising, willsch2022benchmarking}. This disparity highlights the importance of the minor-embedding problem in the context of quantum annealing.

Furthermore, the mapping is not only critical but also highly complex to solve. The challenge arises from the limited connectivity of the processors, making the mapping of any Ising model non-trivial. In fact, the lower the hardware connectivity, the lower the likelihood that a problem can be trivially mapped onto the hardware. For instance, consider an Ising model where a variable interacts with 16 others, while the hardware qubits are only connected to 15 other qubits. This discrepancy hinders a direct one-to-one mapping of variables to qubits, necessitating the use of multiple qubits to represent a single variable.

In this context, solving the minor-embedding problem while minimizing the number of qubits in the embedding is crucial. Using more qubits increases the solution space and, given the noise in the processors and its random nature, it also raises the error rate. However, finding an embedding for any Ising model into any processor while minimizing the number of qubits is an NP-Hard problem. 

Note that the minor-embedding problem in the context of quantum annealing will remain in that complexity class as long as QA processors do not possess all-to-all connectivity. This condition is likely to persist in the future because, due to the 2D nature of superconducting-based processors, it is extremely challenging to construct all-to-all hardware of non-small sizes~\cite{bunyk_architectural_2014}.

Because of the complexity of the minor-embedding problem, heuristic methods have to be employed for its solving. In the literature, the standard algorithm to solve the minor-embedding problem is \textit{Minorminer}. This method, proposed by Cai et al. in 2014 \cite{cai2014practicalheuristicfindinggraph}, uses a greedy approach that iteratively constructs the embedding while minimizing at each step the total number qubits used to represent each variable. 
Minorminer is accessible in the D-Wave problem solving platform \cite{dwave2025} as the standard to embed generic problems into the different available processors. Besides Minorminer, D-Wave offers another minor-embedding method specific to embed fully connected problems, called \textit{Clique Embedding} (CE, \cite{boothby2016fast}). This algorithm gives an embedding in polynomial time for a fully connected instance of a given size.

In addition to the methods provided by D-Wave, various independent research teams have proposed alternative approaches, including the following:
\begin{itemize}
    \item Researchers in \cite{pinilla2019layout} present Layout-Aware Embedding, which utilizes the “location information” of variables within a problem to guide mapping heuristics. This location information can either be an inherent characteristic of the problem or be calculated using graph theory node placement algorithms. Once an initial embedding is achieved through this method, the Minorminer algorithm is employed to optimize it.
    
    \item In \cite{zbinden2020embedding}, two algorithms are proposed to search for an initial embedding, which is subsequently taken as an initial solution by Minorminer. The first method, termed Clique-Based Minorminer, utilizes the CE method by D-Wave to trivially embed a subset of variables, leaving the remaining variables unassigned to qubits. 
    The second one, named Spring-Based Minorminer, is analogous to the layout-aware embedding algorithm presented in \cite{pinilla2019layout}. In a nutshell, it employs a tuned Fruchterman-Reingold spring algorithm to calculate a layout for positioning the variables and interactions of the problem in a plane. This layout information is then used to assign variables to qubits and construct an initial embedding.
    
    \item The team behind \cite{Sugie} present Probabilistic-Swap-Shift-Annealing, which is a minor-embedding algorithm inspired by the principles of simulated annealing. It begins by generating an initial embedding and then iteratively modifies it through probabilistic decision-making to enhance the embedding over successive iterations.

    \item In \cite{goodrich2018optimizing}, the authors introduce the idea of a virtual hardware layer, which consists of a precomputed embedding of a biclique graph into the hardware, making the user calculate embeddings to the biclique instead of to the topology graph. The idea behind this method is that they propose a heuristic algorithm to calculate embeddings of problems into the biclique virtual hardware based on the graph theoretic concept of odd cycle transversal decomposition, which is less costly than embedding instances directly to the hardware.

    \item  In \cite{bernal2020integer}, the minor-embedding problem is reformulated and solved using Integer Programming tools.
\end{itemize}


Despite the high interest of the alternative minor-embedding methods—some of which have demonstrated competitive or even superior performance to Minorminer on specific problem instances, as reported in their respective original studies—Minorminer remains the most widely adopted and standardized approach. This is primarily due to its role as the default embedding algorithm within D-Wave’s Ocean SDK.\footnote{\url{https://docs.dwavequantum.com/en/latest/ocean/api_ref_minorminer/source/index.html}}. Accordingly, our study focuses on Minorminer, as it represents both the state-of-the-art and the most widely used method in practical applications. This choice enables us to assess performance under conditions that are most representative of current usage in the field.

The objective of this work is twofold:
\begin{enumerate}
    \item Reveal the critical influence that the minor-embedding problem has on the final performance of quantum annealing. This is, the correlation between the quality of the embedding (solution to the minor-embedding problem given by Minorminer) and the quality of the final solutions given by the quantum annealer. This publication contains an in-depth study of this influence with both theoretical and experimental results. 
    
    \item Upon recognizing the significance of embedding quality, this work conducts a detailed investigation into the effectiveness of the Minorminer algorithm in producing high-quality embeddings. The study includes a comprehensive experimental evaluation of Minorminer’s performance \textcolor{black}{embedding Erdös-Rényi graphs \cite{erdds1959random}}, with comparisons to the CE method, which is treated as a worst-case baseline due to the assumption that embeddings for fully connected instances should require more qubit demanding embeddings than for instances of lower density. 
\end{enumerate}

The rest of the paper is organized as follows: the subsequent section provides background information.
Section \ref{sec:experimentation} presents the experimental results of the study. Firstly, it evaluates the relationship between the quality of the embedding and the quality of the final solution to the problem. Secondly, it provides an experimental assessment of Minorminer to measure its capabilities and performance. This paper finishes highlighting the main conclusions and further work in Section \ref{sec:conclusions}.

\section{Background}\label{sec:background}
The background information is divided into two parts. Subsection \ref{sec:process} provides a comprehensive overview of the quantum annealing process, detailing the steps required to solve a combinatorial optimization problem using a quantum annealer. This explanation sets the context for the Minor-Embedding problem, which is then thoroughly discussed in Subsection \ref{sec:MEProblem}.

\subsection{The Quantum Annealing Process}
\label{sec:process}

Solving a combinatorial optimization problem with a quantum annealer that samples low energy solutions of programmable Hamiltonians is a process composed of different steps, which are represented in Figure \ref{fig:diagrama_flujo}:

\begin{figure*}[t]
    \centering
    \includegraphics[scale = 0.38]{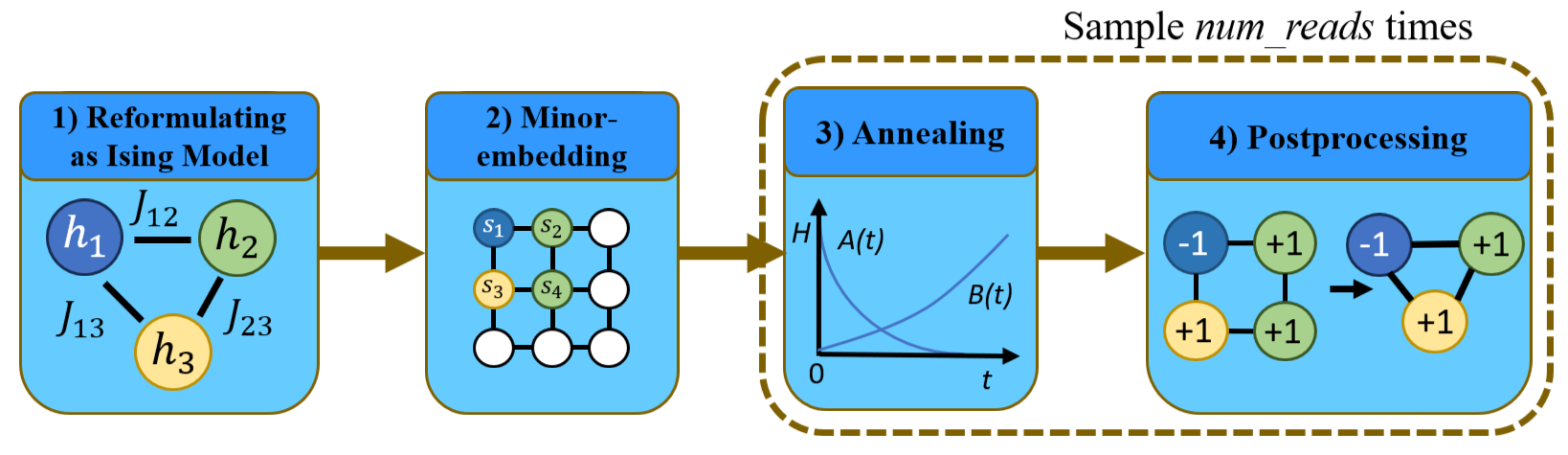}
    \caption{Flux diagram of the quantum annealing problem solving process in D-Wave Systems annealers. The image is an adapted version of the scheme originally presented by Goodrich et. al. \cite{goodrich2018optimizing}.}
    \label{fig:diagrama_flujo}
\end{figure*} 

\begin{enumerate}
    \item \textit{Reformulating the Combinatorial Optimization Problem as an Ising Model}: First, the problem must be modeled as an Ising model. This can be accomplished using various mathematical techniques if the problem involves non-binary discrete variables and/or constraints \cite{lucas2014ising}.
    
    \item \textit{Minor-embedding}: Then the Ising model must be adapted to fit into the quantum annealing processor. These processors consist of qubits and couplers between pairs of qubits, which sample low-energy states of Hamiltonians represented by the following form:
    
    \begin{equation}\label{sampleable_hamiltonians}
        \sum_{i \in V} h_i \sigma_z^{(i)} + \sum_{(i,j) \in E} J_{i,j} \sigma_z^{(i)} \sigma_z^{(j)}
    \end{equation}
    
    where $\sigma_z^{(i)}$ is the Pauli z operator acting on qubit $i$, and $V$ and $E$ are the sets of qubits and couplers in the given processor, respectively. For each existing qubit $i\in V$, $h_i$ is its bias, and for each existing coupler $(i,j)\in E$, $J_{i,j}$ is its weight. After each execution of the annealer, the state of each qubit is measured, resulting in solutions to the following Ising model:
    
    \begin{equation}\label{FO}
        \sum_{i \in V} h_i s_i + \sum_{(i,j) \in E} J_{i,j} s_i s_j
    \end{equation}
    
    where $s_i$ are the variables taking values $+1$ or $-1$ for all $i \in V$.
    
    QAs can be regarded as programmable Ising model solvers. However, their programmability is constrained by the characteristics of the processors. The values $h_i$ and $J_{i,j}$ can only be adjusted within a certain range, and the sets $V$ and $E$, which define the possible binary variables and their interactions, respectively, are dependent on the architecture of the processor.
    
    By considering the set of qubits $V$ as the node set and the set of couplers $E$ as the edge set, the processor's architecture is represented by the \textit{topology graph} $G = (V, E)$. The Ising model can also be represented by the \textit{problem graph}, where the nodes correspond to the variables and the edges represent the non-zero quadratic terms that connect pairs of variables. Representing both the problem and the hardware as graphs is crucial because the Ising models that can be solved on these machines are those whose problem graphs are subgraphs of the processor's topology. However, due to connectivity limitations in hardware, it is highly unlikely that the problem graph of an Ising model for a real combinatorial optimization problem is naturally a subgraph of any current quantum annealer’s topology.
    
    In this context, addressing the minor-embedding problem is required, which is done by intelligently selecting qubits or sets of qubits to represent the variables of the Ising model. The latter case is occasionally necessary, as this approach allows for the representation of variables with higher connectivity than the degree of the topology graph, thereby enabling the representation of an exponentially large number of Ising models. Historically, the sets of qubits assigned to single variables are termed \textit{chains}. However, these chains are not strictly linear paths of qubits; they can instead form arbitrarily connected subgraphs.
    
    Remark that there is not a unique way of embedding a given problem graph into a given hardware graph, and it is crucial to find the one which employs the least amount of qubits possible. This is because an increase in the number of qubits expands the solution space that the annealer must explore, which, due to imperfections in current processors, leads to a degradation in solution quality.

    Upon solving the minor-embedding problem, the original Ising model is transformed into the \textit{embedded Ising model}. The graph of this new model is a minor of the hardware's topology graph, and penalties are introduced to promote that qubits forming chains attain the same value at the end of the anneal. The strength of these penalties is adjusted using the \textit{chain strength} parameter, and selecting an appropriate value for this strength is critical: values that are too low do not ensure the mentioned qubit value stability through the chains, while values that are too high create a large energy gap between penalized and unpenalized solutions, compared to the difference between good and bad solutions of the original problem, making it difficult for the annealer to find the optimal solutions \cite{choi2008minor, king2014algorithm, raymond2020improving}. In the D-Wave paradigm the chain strength is determined through the \textit{Uniform Torque Compensation} (UTC) function, first introduced in \cite{raymond2020improving}. This function operates with a user-defined prefactor, which has a default value of $1.414$. Anyway, this is a problem-dependent parameter, so it is recommended to be empirically defined for each instance and embedding within the $[0.5, 2]$ range \cite{dwave2025, hamerly2019experimental}.
    
    Additionally, variable biases and interaction weights are adjusted to promote that the solution space encompasses that of the original problem. Consequently, the embedded Ising model can be solved via quantum annealing, and its solutions can be efficiently translated back to the original problem.
    
    \item \textit{Quantum Annealing}: This is the step executed by the quantum annealing processor. The quantum annealer is initialized in a physically easy-to-prepare state, and its Hamiltonian is gradually modified until it matches the problem's Hamiltonian in Equation \eqref{sampleable_hamiltonians}, where the embedded Ising model is encoded. At the end of the evolution, each qubit's state is measured to obtain a solution to the embedded Ising model. Due to imperfections in current annealers and the probabilistic nature of the method, the annealer is run multiple times to generate a sample of solutions. The number of executions, referred to as \textit{number of reads}, is determined by the user; however, it is recommended to perform at least 1000 executions \cite{dwave2025}.
    
    \item \textit{Postprocessing}: The sample must be postprocessed to get solutions to the Ising model. This process involves taking the final value of each qubit to determine the value of the corresponding Ising variable. When a variable is represented by a chain, two scenarios are possible: all qubits in the chain have taken the same value (i.e. there is a consensus), or qubits in the chain have differed in the readout of the sampling. The latter scenario is referred to as a \textit{broken chain}. In D-Wave's quantum annealing paradigm, a broken chain is resolved through a majority vote among the qubits in the chain, assigning the winning value to the variable. Lastly, the solution to the Ising model is transformed to be of the form of the original combinatorial optimization problem.
    
\end{enumerate}

Note that the quantum annealing approach described above is not applicable to problems of all sizes. When the Ising model exceeds a specific size and connectivity threshold, finding a minor-embedding into the topology graph becomes infeasible. In such cases, quantum annealing based hybrid computing methods, such as those offered by D-Wave \cite{HSS}, can be employed. However, this research area falls beyond the scope of this paper, as it involves additional classical methods beyond minor-embedding.

\subsection{The Minor-Embedding Problem}
\label{sec:MEProblem}

The minor-embedding problem emerged in the field of graph theory well before the advent of quantum computing. This problem has been examined in numerous publications that explore its complexity and propose algorithms, such as those by Matoušek (1992) \cite{matouvsek1992complexity} and Robertson and Seymour (1995) \cite{robertson1995graph} (the latter under the name of the minor containment problem). Later, Vicky Choi (2011) \cite{choi2011minor} redefined the problem specifically for quantum annealing applications.

As mentioned in the previous section, minimizing the number of qubits is crucial when embedding Ising models into quantum annealing hardware. Therefore, in this work, we address the Minor-Embedding problem not only in the conventional sense of finding an embedding, but rather as an optimization problem aimed at identifying the minimal embedding in terms of qubit usage. We refer to this as the Minimal Minor-Embedding Problem, which's definition is the following. In terms of notation, we represent the node set and edge set of a graph $G$ as $V(G)$ and $E(G)$, respectively, and the power set of a set $A$, which is the set of all its subsets, as $\mathcal{P}(A)$.

\begin{definition}{}{def}
  Given a pair of graphs, called source graph $H$ and target graph $G$, the minimal minor-embedding problem is the problem of finding the minimum minor of $G$ that, by contraction of connected nodes, transforms into a graph isomorphic to $H$.
  
  The solution to the minor-embedding problem is an \textit{embedding} or \textit{model}, which defines as a function $\phi: V(H) \to \mathcal{P}(V(G))$ subject to the following conditions:
  \begin{enumerate}
    \item \textbf{Vertex-model connectivity:} For every node $v \in V(H)$, $\phi(v) \subseteq V(G)$ forms a connected subgraph in $G$.
    \item \textbf{Edge representation:} For every edge $(u, v) \in E(H)$, there exists one edge $(u', v') \in E(G)$ where $u' \in \phi(u)$ and $v' \in \phi(v)$.
    \item \textbf{Vertex-model pairwise disjointness:} $\phi(u)\cap \phi(v) = \emptyset $ for all $u, v \in V(H)$ such that $u\ne v$. This is, every node in the target graph can only appear in the mapping of one source graph node.
  \end{enumerate}
  The set $\phi(v)$ is called the \textit{vertex-model} or \textit{chain} of $v$.
\end{definition}

In the context of embedding Ising models into quantum annealers, the source graph represents the problem graph, while the target graph corresponds to the architecture of the quantum processor. The node set $V(H)$ in the source graph comprises the variables of the model, while the edge set $E(H)$ consists of pairs of variables that exhibit non-zero weighted quadratic interactions. Conversely, the target graph is defined with qubits as nodes and couplers as edges. The minimum minor of the target graph refers to the minimal utilization of qubits.

At the time of writing, the most advanced quantum annealers are D-Wave's \texttt{Advantage\_system} processors. These processors are constructed using a topology graph known as \textit{Pegasus} \cite{boothby2020next}, illustrated in Figure \ref{fig:pegasus}. This graph comprises 5760 nodes, with the majority having a degree of 15. In addition to the established \texttt{Advantage\_system} processors, a prototype of the next-generation models is also available, which is called \texttt{Advantage2\_prototype}. This new device features a higher connectivity topology known as \textit{Zephyr} \cite{boothby_zephyr_nodate}. This new topology has a degree of 20 for the majority of the nodes. However, due to its early stage, this prototype is not as large as the established ones and features approximately 1200 qubits.

\begin{figure}[t]
    \centering
    \includegraphics[width=0.7\linewidth]{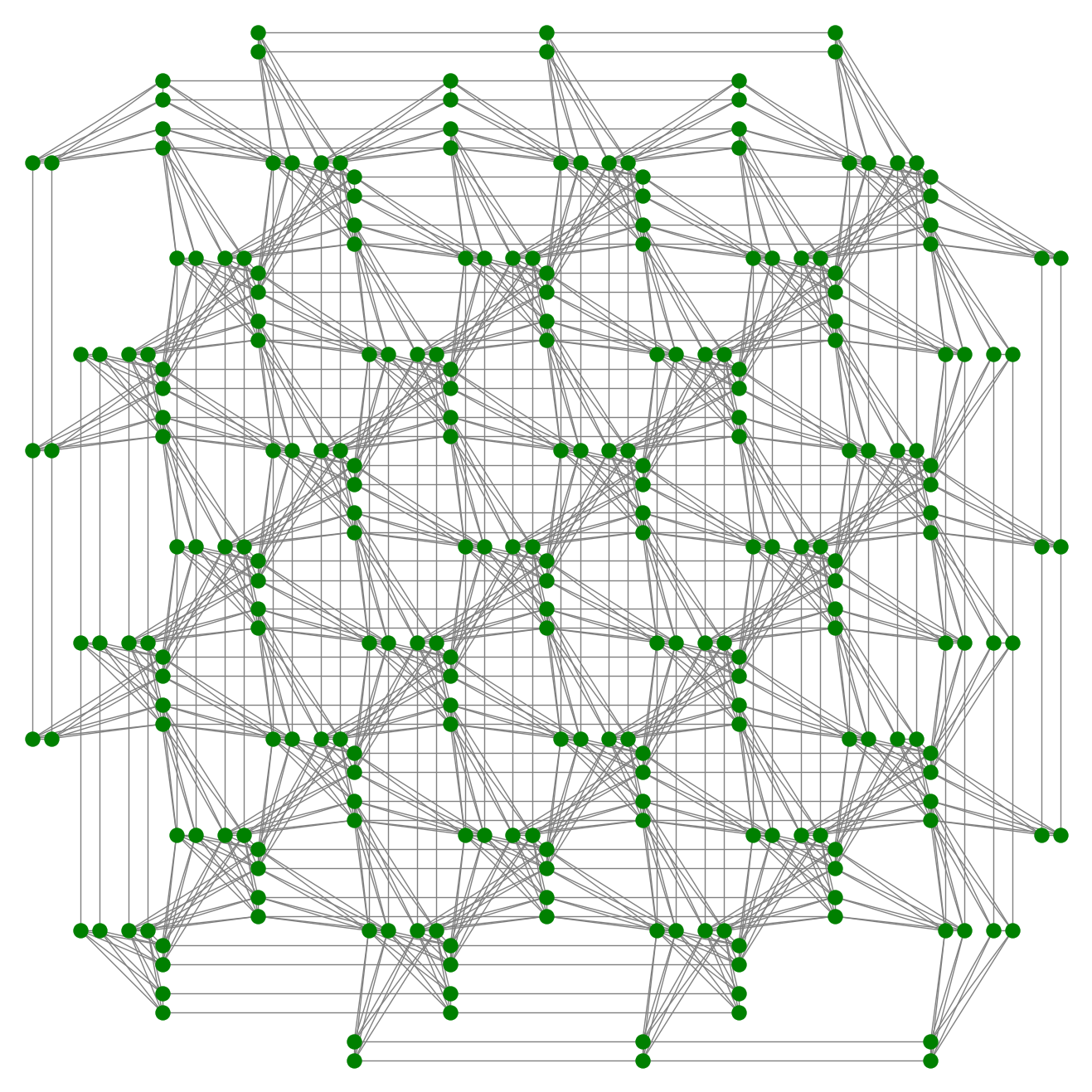}
    \caption{Section of a Pegasus graph. The size of the actual Pegasus employed in the processors follows the same shape but is 16 times larger in area. The graph has been generated with the D-Wave's python library dwave-networkx, which is an extension of Networkx \cite{hagberg2008exploring}.}
    \label{fig:pegasus}
\end{figure}

Although the ideal Pegasus graph is the one shown in Figure \ref{fig:pegasus}, in practice, certain qubits and couplers are excluded from the accessible programmable fabric due to operational malfunctions. The actual topology graph of an \texttt{Advantage\_system} processor, which includes missing qubits or couplers, is referred to as a \textit{broken} Pegasus graph, and the actual number of qubits of these graphs is around 5600. The number of couplers in the \texttt{Advantage\_system} processors is approximately 40\,000, resulting in a sparse graph with a density of
\[\frac{\text{existing couplers}}{\text{number of possible pairs of qubits}}\approx\frac{40\,000}{\displaystyle \binom{5600}{2}} \approx 0.0025.\]
This sparsity significantly limits the capabilities of these annealers. 

\begin{Example}[title={\bfseries Illustrative example}]{}{example}
To illustrate the impact of the sparse density of quantum annealing processors, consider an example using the D-Wave \texttt{Advantage\_system4.1}. Take an instance of the well-known Traveling Salesman Problem (TSP) with $n$ cities. When formulated as an Ising model, it requires \( n^2 \) variables, resulting in a fully connected problem graph \cite{lucas2014ising}. Embedding a complete graph into Pegasus, due to its limited connectivity, necessitates extensive chains, thereby requiring a large number of qubits. Consequently, only fully connected Ising models with fewer than 200 binary variables can be embedded, thereby limiting the TSP instance to fewer than \( \sqrt{200} \approx 14 \) cities.
\end{Example}

To summarize, solving the minor-embedding problem optimally is crucial for the following reasons:
\begin{itemize}
    \item The performance of the end-to-end quantum annealing paradigm is influenced by the quality of the embedding, as will be demonstrated in Section \ref{sec:embedding_quality_affection}.
    \item The size and density of the problems the quantum annealer can address depend on the ability to embed problems using the minimum number of qubits.
\end{itemize}

However, the relevance of the problem is further underscored by its computational hardness. Note that the following hardness result pertains to the decision problem associated with our optimization problem, which involves determining whether a minor-embedding can be achieved using the minimal possible number of qubits.

\begin{proposition}{}{np-hard}
    The minimal minor-embedding problem is NP-hard.
\end{proposition}
\begin{proof}
    As shown in Lemma 3 from \cite{LOBE2024114369}, the problem of just finding an embedding without minimizing the size of the minor (the number of qubits) is NP-complete. This is because for any given target graph, the special case of taking as source a cycle graph with the same number of vertices as the target, corresponds exactly to the Hamiltonian Cycle problem, which is one of Karp's 21 NP-complete problems \cite{Karp1972}. Also, its solution is polynomially verifiable by checking that all three conditions of the definition are fulfilled.

    However, the minor-embedding problem as defined here is NP-hard because it is at least as hard as the NP-complete model, but no polynomial algorithm is known to verify the optimality of a given solution in terms of used number of qubits.
\end{proof}

It is important to note that the previous result is proven for the problem with source and target graphs as inputs. In the quantum annealing paradigm, this is not exactly the case because the target graph is not fully variant; it is restricted to slightly broken versions of specific topologies such as Pegasus, Zephyr, or other topology graphs of different quantum annealers. However, as demonstrated by Lobe and Lutz in \cite{LOBE2024114369}, the variants with missing vertices and edges in these broken topology graphs maintain the problem within the same complexity class.

As a result of the problem's complexity, heuristic methods must be employed, which have inherent disadvantages. These include suboptimal embeddings that introduce unnecessary noise and the inability to embed problems that could theoretically be embedded, further diminishing the quantum annealer’s capabilities. Moreover, finding the actual minor-embedding typically demands significant computational time. The classical solving of the minor-embedding problem can delay the quantum annealing process, potentially negating any speedup advantages.

\section{Experimentation}\label{sec:experimentation}
Given the comprehensive background provided on the minor-embedding problem in quantum annealing, we now present the core findings of this study. This section aims to experimentally address the following closely interrelated research questions (RQs). 

\begin{quote}
    \textit{\textbf{RQ1}: What is the impact of the embedding quality on the performance of the quantum annealer? (Section \ref{sec:embedding_quality_affection})}
\end{quote}

QA's performance is influenced by different factors such as noise in the processors and the nature of the energy landscape of the Ising model. In this context, we hypothesize that the error in quantum annealing \textcolor{black}{depends critically on the quality of the embedding.}


Noise is naturally unpredictable, and the energy landscape can only be assessed once the problem is analytically solved, making preemptive adjustments unfeasible. However, the number of qubits in the embedding is known prior to sending the problem to the processor. Thus, although not a holistic performance predictor, our experiments offer insights into the annealer’s problem-solving efficacy based on the embedded problem’s size.

\begin{quote}
    \textit{\textbf{RQ2}: How good is Minorminer as a minor-embedding algorithm for quantum annealing? (Section \ref{sec:MMPerformance})}
\end{quote}

Given that Minorminer is the standard method for embedding problems into D-Wave’s annealers, evaluating its performance is crucial. This analysis provides insights about the expected quality of embeddings of a given problem. When combined with the previous point, it helps predict the anticipated performance of the end-to-end quantum annealing solution. \textcolor{black}{Note that the experiments were conducted on a specific class of random graphs which, although representative of a wide range of problems, do not encompass all possible problem structures.}

Section \ref{sec:setup} will provide the detailed configuration of the experiments to ensure the reproducibility of the results.

\subsection{Experimental setup}\label{sec:setup}
Three experiments were conducted in this study: two addressing RQ1 and one addressing RQ2.

\paragraph{Instance generation}
The problems to be solved along the whole experimentation are Ising models which have been generated following a two-step procedure: 
\begin{itemize}
    \item Initially, the problem graph is randomly constructed using the Erdös-Rényi (ER) algorithm \cite{erdds1959random}, tailored to specific problem sizes and densities. The algorithm was selected because its uniformly random nature makes it a generally applicable and relatively unbiased method for approximating graphs with a given number of nodes and edge probability (density).
    \item Subsequently, random numerical values are assigned to the variable biases at the nodes and the interaction strengths at the edges. These values are uniformly selected from the range $[-1,+1]\subset \mathbb{R}$ to prevent degenerate instances. 
\end{itemize}
The size, density, and number of problems generated and solved varied across the experiments, and is detailed in Table \ref{table:setup}.

\paragraph{Infrastructure}

The embeddings for all experiments were conducted on Katea Computing Platform, the high-performance computing infrastructure at Tecnalia. The system features 1056 CPU cores operating at 3.8 GHz, providing approximately 79.8 teraflops of computational power, complemented by 9.5 TB of RAM and 1 petabyte of storage. Code was run with Python 3.9.18. Specifically, the experiments were executed using 64 CPU cores allocated from the CPU partition via the SLURM workload manager.

All Ising models were solved using the D-Wave \texttt{Advantage\_system4.1} processor, which, at the time of writing, is the largest model available (the one with least amount of broken qubits), featuring 5627 qubits arranged in a broken Pegasus topology. The processor has been operated with default parameters, and the number of reads per annealing process has been set to 1000 reads.

Minorminer has been executed differently for each of the two research questions:
\begin{itemize}
    \item For RQ1, the chain-length-patience parameter—which determines the number of consecutive iterations without improvement in qubit usage before termination (see \cite{cai2014practicalheuristicfindinggraph})—was configured differently across executions (see Table~\ref{table:setup}).  This setup has been chosen to obtain a broader and more varied sample, allowing for the observation of the relationship between QA performance and embedding qubit usage per variable.
    \item For RQ2, the chain-length-patience parameter has been set to its default value (10), as the goal has been to assess the algorithm's capacity.
\end{itemize}

Regarding the chain strength parameter, in RQ1.1 it has been calculated using D-Wave's UTC function, with the prefactor set to the default value 1.414. While empirical tuning of the prefactor is recommended to optimize performance, the default setting was employed in this case to isolate and evaluate the impact of the embedding on solution quality. This choice also reflects the typical usage patterns of non-expert users, who are likely to rely on default configurations. 
In contrast, for RQ1.2, each problem and embedding was executed using five different prefactor values ranging from 0.5 to 2.0, consistent with D-Wave’s recommended tuning range, and their results have been studied.

The Minorminer's chain-length-patience parameter and Uniform Torque Compensation's prefactor used in the experiments are also detailed in Table \ref{table:setup}.

\begin{table*}[t]
\centering
\renewcommand{\arraystretch}{1.5}
\begin{threeparttable}
\begin{tabular}{lcccccc}
\textbf{Experiment} & \textbf{Sizes} & \textbf{Densities} & \textbf{Problems\tnote{1}} & \textbf{Embeddings\tnote{2}} & \textbf{CLP\tnote{3}} & \textbf{UTC pref.\tnote{4}} \\
\hline
RQ1.1 & $\{25, 35, ..., 175\}$ & $\{0.05, 0.1, ...,1\}$ & 5 & 10 & \{0,1,...,10\} & 1.414 \\
RQ1.2 & 150 & 0.5 & 5 & 100 & \{0,1,...,10\} & \{0.5, 0.75, 1, 1.414, 2\} \\
RQ2 & $\{10, 15, ..., 300\}$ & $\{0.05, 0.1, ...,1\}$ & 1 & 64 & 10 & \_\\
\end{tabular}
\begin{tablenotes}
    \footnotesize
    \item[1] Number of problems per density and size.
    \item[2] Number of embeddings per problem.
    \item[3] CLP stands for Minorminer's chain-length-patience parameter.
    \item[4] Chain Strength setting function Uniform Torque Compensation prefactor.
\end{tablenotes}
\end{threeparttable}
\caption{Experimental setup for the three experiments. 
}
\label{table:setup}
\end{table*}

\paragraph{Evaluation procedure}
As D-Wave quantum annealers work as sampling machines, it is more pertinent to consider the median energy of the sample rather than the solution with the minimum energy to accurately assess the machine’s performance. While the best-found solution is obviously more interesting when solving a problem, the median is arguably more representative of the entire sample. Therefore, the results presented in this study are based on the median value of the sample.

Furthermore, since the annealer tackles the embedded problem, it is more pertinent to evaluate the solutions to the embedded problem (after step $3$ in Figure \ref{fig:diagrama_flujo}) rather than the postprocessed ones (after step $4$ in Figure \ref{fig:diagrama_flujo}), which can be altered by classical methods. Nevertheless, as the energies of both the solutions to the embedded problem and the postprocessed solutions fall within the same range, they can be shown side by side. Therefore, in Section \ref{sec:embedding_quality_affection} both sets of solutions are studied.

\paragraph{Metrics}
Regarding the metrics employed to asses the quality of the results, on the one hand, the solutions provided by the QPU are measured using the relative error. This metric quantifies the ratio of the deviation of the energy of the solution with respect to the energy of a reference solution. The reference solutions used along the experimentation have been taken as the best from the two provided by the D-Wave implemented versions of the following classical algorithms: MST2 multistart tabu search algorithm \cite{palubeckis2004multistart} and simulated annealing sampler \cite{kirkpatrick1983optimization}. Thus, the relative error is calculated as follows:
\[
e_{rel} = \left|\frac{energy_{ref}-energy_{QA}}{energy_{ref}} \right|
\]
where $energy_{ref}$ is the energy of the solution found by the reference solver and $energy_{QA}$ is the energy of a solution given by the annealer, i.e., the energy of an outcome after step $3$ in Figure \ref{fig:diagrama_flujo}. This approach is justified despite the reference solutions being obtained via heuristic methods, as these solvers have demonstrated strong performance on problem types similar to those addressed in this work \cite{mcgeoch2013experimental, gilbert2024benchmarking, tasseff2024emerging}. Moreover, in the experiments conducted here, no solution produced by quantum annealing outperformed those generated by the reference solvers.

On the other hand, the quality of an embedding can be evaluated using two different but related metrics. The first one is the total number of qubits needed in the embedding, i.e.,
\[
    n_{\mathrm{qubits}} = \sum\limits_{v_i\in V(H)}|\phi(v_i)|,
\]
and the second one is the \textit{average chain length (ACL)}:
\[
ACL = \frac{n_{\mathrm{qubits}}}{size(model)}
\]
where $size(model)$ is the number of variables in the original Ising model. As the size of the Ising models varies throughout the experiments, we use the latter metric because it is normalized by the number of qubits, making it more representative for comparing embeddings of Ising models with different sizes.
 
The comparison between QA performance and embedding quality is therefore quantified by the median relative error of the embedded problem solution sample, relative to the ACL of the embedding used to solve the problem.

In relation to RQ1.2, where the chain strength was varied across executions, the chain break fraction was also examined. This metric, defined as the ratio of broken chains in a solution to the total number of variables in the problem instance, serves as an indicator of the effectiveness of the chosen chain strength in maintaining chain integrity throughout the solutions.

All experiments described in this study are fully reproducible and openly available\footnote{https://github.com/aitorgtt/Addressing-the-MEP-in-QA-and-evaluating-SOTA-algorithm}.

\subsection{RQ1: What is the impact of the embedding quality on the performance of the quantum annealer}
\label{sec:embedding_quality_affection}

In general, problems with larger solution spaces are more challenging to solve than those with smaller ones. As previously discussed, optimal embeddings minimize the number of qubits used, while poor embeddings result in larger embedded problems. Even when the embedded problem is correctly modeled, an increase in its size expands the solution space within which the quantum annealer operates, thereby reducing the probability of obtaining high-quality solutions. The growth of the solution space is analytically measured in the following proposition:

\begin{proposition}{}{sol_space}
    Let \( n\in\mathbb{N} \) denote the number of variables in an Ising model, and let the model be embedded onto a hardware topology graph with an average chain length denoted by ACL. Define the solution spaces of the original and embedded Ising models as \( S_{original} \) and \( S_{embedded} \), respectively. Then, the proportion of solutions in \( S_{embedded} \) that corresponds to solutions without any broken chains—i.e., valid mappings to \( S_{original} \)—is given by:
    \[
    \frac{|S_{original}|}{|S_{embedded}|}=2^{n (1 - \mathrm{ACL})} \to 0 .
    \]
    exponentially fast as ACL increases. 
\end{proposition}




\textcolor{black}{In addition to the fact that the error of quantum annealers increases linearly with the number of qubits, as reported in \cite{raymond2016global, albash2017temperature}, in the context of embedded Ising models the error grows even more rapidly. This is due to the contraction of the space of valid solutions and to the noise introduced by the chain strengths in the final objective function, even when this parameter is optimally tuned.}



To investigate the relationship between embedding quality and quantum annealing performance, two experiments have been conducted.

\paragraph{RQ1.1: Embedding Impact - General Case}
The first experiment measures the relationship between embedding quality and QA performance for a set of 400 Ising models with varying sizes and densities across the entire embeddable range (detailed in Section \ref{sec:MMPerformance}). The aim is to assess this relationship across different scenarios.
The results are represented in Figure~\ref{fig:exp_regresion}.


\begin{figure}[t]
    \centering
    \includegraphics[scale = 0.6]{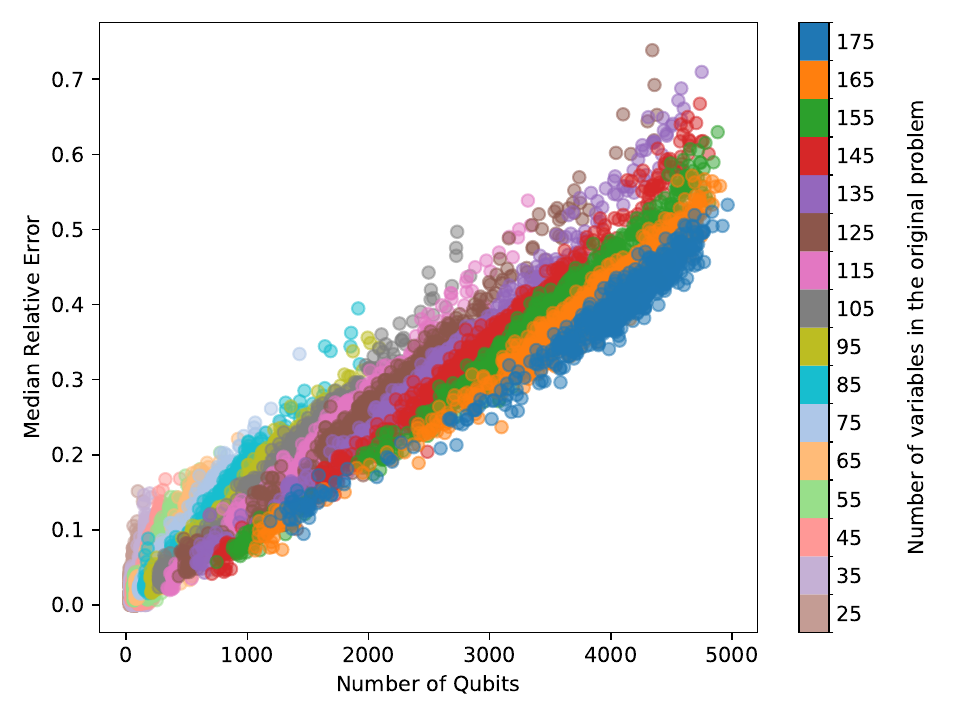}
    \caption{\textcolor{black}{The colored dots represent the median relative errors over the Number of Qubits of the embedded Ising Models for each solution sample. These samples are derived from solving randomly generated Ising models as described in Table \ref{table:setup}.}} 
    \label{fig:exp_regresion}
\end{figure} 


\begin{figure*}[t]
    \centering
    \includegraphics[width = 1\linewidth]{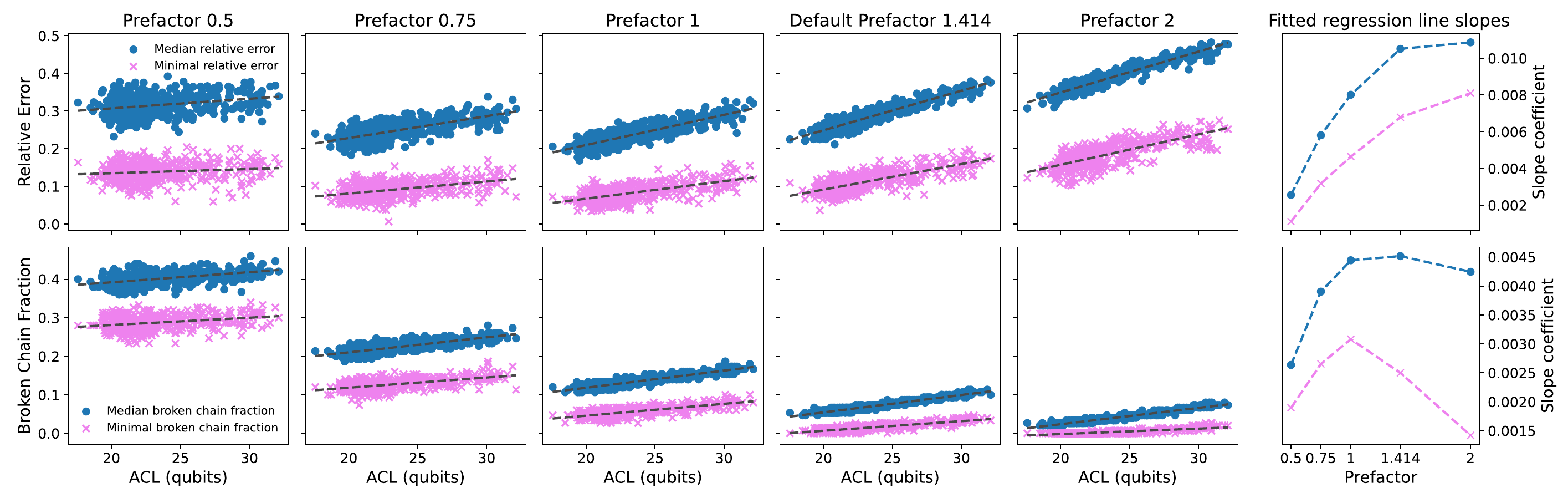}
    \caption{\textcolor{black}{Top row, first five panels: Median and minimum relative errors of the post‑processed solutions obtained from 100 samples, with fitted linear regressions shown as dashed lines. Each point corresponds to a complete sample generated using a distinct embedding, for five randomly selected problem instances (see Table \ref{table:setup}). Bottom row, first five panels: Median and minimum broken‑chain fractions plotted as a function of the ACL of the embeddings, with fitted linear regressions. Last column: Slopes of the linear regression models fitted for each chain‑strength prefactor. For every problem instance and embedding, five embedded Ising models were produced by varying the chain strength, controlled by the prefactor in the UTC function.}}
    \label{fig:RQ1.2}
\end{figure*}

\textcolor{black}{The most immediate observation from the figure is the expected linear relationship between the total number of qubits and the resulting error, consistent with previous findings in the literature. However, an additional pattern becomes evident when the color grouping is examined: for a fixed total qubit count, instances derived from smaller original problems exhibit higher error. In other words, larger problem instances that require comparatively simpler embeddings tend to yield more accurate solutions than smaller problems whose embeddings demand a higher qubit overhead.}

\textcolor{black}{This behavior is reasonable. Embeddings that use more qubits produce a smaller proportion of valid configurations, and they impose more stringent constraints on the selection of chain strengths. As will be shown in the following section, it is not feasible in practice to choose chain strengths that fully prevent chain breaks and still retain the best possible solution quality. Consequently, the error increases inevitably, and this highlights the critical importance of obtaining high‑quality embeddings.}



\paragraph{RQ1.2: Embedding Impact - Stressed} 
The second experiment analyzes the relationship between QA performance and embedding quality under fixed problem density and size conditions near the embeddability limits (as discussed in Section~\ref{sec:MMPerformance}). This setting is intentionally chosen to represent non-trivial, challenging instances and to stress the embedding process. \textcolor{black}{However test were also made on small and medium instances, and these results can be found in \ref{app:small_medium}.} Since the problem density and size remain constant, variations in QA performance are primarily attributed to differences in the quality of embeddings—captured through the average chain length. 

To assess whether the hypothesized relationship between embedding quality and QA performance holds independently of other factors—and under optimal parameterization of the embedded Ising models—this experiment varies the chain strength across five distinct values for each embedding by adjusting the prefactor in the chain‑strength calculation of the UTC function. This procedure yields five different embedded Ising models per problem instance. The objective is to characterize how performance correlates with ACL across these different chain strengths and to determine whether this relationship persists at the optimal chain strength for each embedding. As a result of this setup, direct comparison of raw energies becomes less meaningful due to differing penalty scales; therefore, we focus on postprocessed solutions—the ones users ultimately receive—and use blue and violet to represent them in Figure~\ref{fig:RQ1.2}. 

\textcolor{black}{To examine the relationships between the observed variables in greater depth—and given the clear correlations present in the data—linear regression models have been fitted and the slope of each line for every dataset has been analyzed. This approach makes it possible to identify which chain strengths lead to higher or lower relative errors and broken‑chain fractions, as well as to understand how the relationships between these variables evolve as a function of the chain strength.}

As shown in Figure \ref{fig:RQ1.2}, selecting an appropriate chain strength requires a trade-off, and as proven by studies such as \cite{dwave2025, hamerly2019experimental} is problem-dependent. That being said, and although the study of chain strength parameterization is not the main objective of this work, we briefly analyze the results obtained for these problem instances, as they may offer some relevant insights:

\begin{itemize}
    \item The highest tested strength (prefactor 2) minimizes broken chains but degrades solution quality due to excessive penalization. In contrast, the lowest strength (prefactor 0.5) yields lower-energy solutions but distorts the energy landscape, increasing reliance on postprocessing and reducing sampling effectiveness. In such cases, the sampling behavior approaches that of random sampling, which fails to consistently produce high-quality solutions.
    
    \item \textcolor{black}{The optimal chain strength is 1 for nearly all problem–embedding combinations, with only a few cases—generally corresponding to embeddings with higher ACL—showing an optimum at 0.75. This suggests that suboptimal embeddings require a slightly weaker chain‑strength prefactor, accepting a higher number of broken chains in exchange for improved performance. However, it is important to note that these solutions still yield lower overall quality than those obtained from embeddings with lower ACL when using a prefactor of 1.}
    
    
    \item There is a positive correlation between ACL and the fraction of broken chains. As expected from Proposition 2, higher qubit needed embeddings increase the likelihood of encountering broken chains in the solution space, thereby raising the probability that the sampler returns such solutions.

    \item \textcolor{black}{The slopes of the fitted linear models for both the median and minimum relative errors increase with chain strength. At low strengths, sampling is almost random, so the slopes remain small and errors barely depend on ACL. As chain strength grows, sampling improves and poor ACL has a stronger effect, causing the slopes to rise sharply. At the highest strengths, this growth slows down because the impact of ACL begins to saturate. In contrast, the slopes for the broken‑chain fractions also rise at low strengths for the same randomness‑related reason, but they decrease once the prefactor—and therefore the chain strength—becomes larger. This decrease is more pronounced for the minimum values, since stronger chains break less often and thus the broken‑chain fraction grows more slowly with increasing ACL.}
\end{itemize}

All these insights suggest that, when the chain strength is properly tuned, the correlation observed in Figure~\ref{fig:exp_regresion} holds: there is a clear and well-defined upward trend in the relative error (most notably in the median, but also in the minimum values of the solution samples) as the ACL increases. This trend likely stems from the decreasing probability of sampling valid solutions as ACL grows, which inherently leads to a degradation in solution quality. Notably, this correlation is present across all tested chain strengths, with the exception of the weakest settings, where the increased influence of randomness obscures the trend.

\textcolor{black}{}

\begin{Summary}[title={\bfseries Summary RQ1}]{}{firstsummary}

Increasing the Average Chain Length (ACL) in embeddings causes an exponential drop in the fraction of valid solutions—i.e., those without broken chains—within the embedded Ising model's solution space. This impairs sample quality, even with optimal chain strength.
Empirical results show a strong negative correlation between ACL and the performance of the quantum annealing process, \textcolor{black}{when either the size of the Ising Model or the number of qubits are held fixed}. These findings highlight the essential role of optimal minor-embedding, as embedding quality directly governs the overall effectiveness of quantum annealing.
\end{Summary}

\subsection{RQ2: How good is Minorminer as a minor-embedding algorithm for quantum annealing?}
\label{sec:MMPerformance}

To evaluate the performance and capabilities of Minorminer as a minor-embedding algorithm \textcolor{black}{for ER graphs into a broken Pegasus}, this section will examine four key aspects: 
\begin{itemize}
    \item the ability to obtain valid embeddings,
    \item the quality of these embeddings,
    \item the dispersion of the quality of embeddings,
    \item and the execution times.
\end{itemize} 

All four aspects have been studied based on the results of an experiment that involves repeatedly executing Minorminer to embed different problem graphs as specified in Table \ref{table:setup}. CE has also been used to solve the same problems, providing a baseline for comparison with Minorminer. CE identifies embeddings for fully connected graphs of a specified size within a target graph. These embeddings can then be applied to graphs of any connectivity with the same number of nodes, making CE useful as a worst-case scenario in terms of qubit usage (ACL). This algorithm is preprocessed once for the given target graph and subsequently generates embeddings in a deterministic way for a given size of complete graphs in polynomial time, which is significantly faster than Minorminer. The size of the maximum embeddable complete graph is contingent upon the topology graph. For the D-Wave \texttt{Advantage\_system4.1} processor's broken Pegasus, this value is 177 nodes. The ACL of these embeddings is always the same for a given complete graph size, which increases linearly with the number of nodes in the instances.

\paragraph{Embeddableness boundary} The primary objective is to delineate the boundary between instance graphs that can be embedded into D-Wave’s processor topology graphs and those that cannot. Although this boundary is inherently a property of the topology graph, it is also affected by the limitations of the heuristic algorithm employed for the embeddings. Due the intractable nature of the problem, it is not feasible to determine this boundary precisely. Consequently, in practical terms, this boundary can be considered a characteristic of the heuristic method employed. As the results are inherently variable, the boundary is measured by the probability of obtaining a valid embedding.

\begin{figure}[t]
    \centering
    \includegraphics[scale = 0.4]{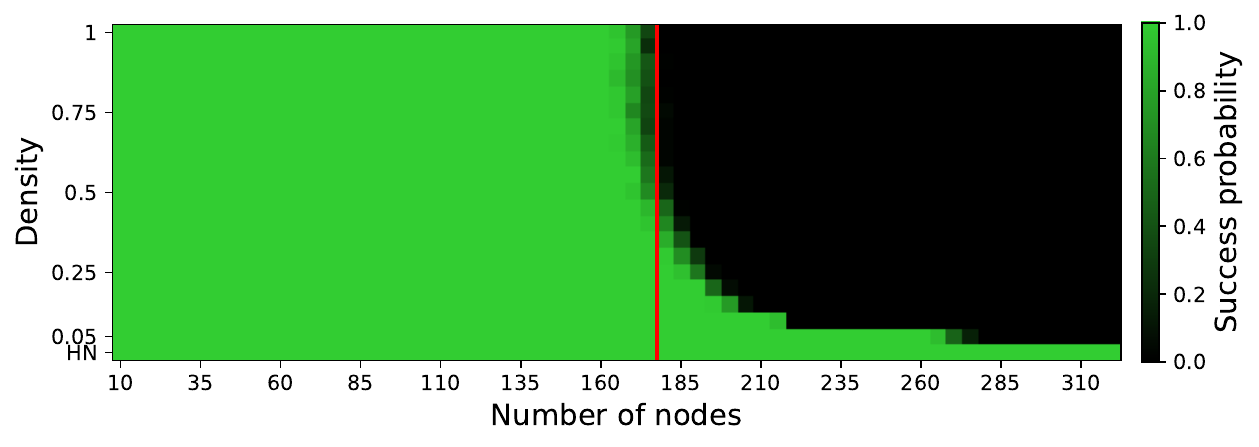}
    \caption{Probability of success of Minorminer finding a valid embedding for ER graphs into D-Wave \texttt{Advantage\_system4.1}'s broken Pegasus, as a function of density and size of problem graph. The red line represents the embeddable limit for CE. HN corresponds to Hardware Native instances, taken as the lowest significant density value.}
    \label{fig:er_frontier}    
\end{figure}

Figure \ref{fig:er_frontier} illustrates the probability of Minorminer of finding a valid embedding as a function of the problem graph's density and size. This probability has been determined by counting the number of valid embeddings found across 64 executions per graph. The lowest density value of the figure corresponds to hardware native graphs, which are trivially embeddable for sizes up to 5600. It is important to note here that Minorminer itself does not verify whether a given problem instance is natively embeddable on the quantum hardware. However, the \texttt{AutoEmbeddingComposite} class in the D-Wave Ocean SDK\footnote{\url{https://dwave-systemdocs.readthedocs.io/en/latest/reference/generated/dwave.system.composites.AutoEmbeddingComposite.sample.html}}, which automates the quantum annealing process explained in Section \ref{sec:process}, initially attempts to submit the problem directly to the quantum annealer under the assumption that it may be hardware native. If this direct submission fails, the composite then invokes Minorminer to compute a suitable embedding.

The results highlight the significant constraints in solving non-sparse instances using D-Wave quantum annealers. The insights drawn from the example in Section \ref{sec:MEProblem} are confirmed: the combination of the hardware’s limited connectivity, the complexity of the minor-embedding problem, and the heuristic nature of Minorminer severely restricts the size of non-sparse embeddable instances. The following are the key observations to be made from Figure \ref{fig:er_frontier}:
\begin{itemize}
    \item Despite the hardware containing over 5600 qubits, the size of ER graphs with a density of 0.05 is limited to approximately 275. When the density exceeds 0.5, this value decreases to 175. 
    \item Also, note that the boundary between embeddable and non-embeddable graphs is not a simple linear function of the problem's size or density. The size of the largest embeddable instances decreases rapidly at lower density values, but this decline becomes less pronounced as the density increases. In fact, the curve becomes almost vertical, suggesting that the maximum embeddable instances for fully connected graphs and those with a 0.5 density are of the same size, even though the latter possesses only half the interactions of the former. 

    A possible explanation to this phenomenon is attributed to the constraints imposed by embedding graphs that are significantly denser than the topology graph. For density values exceeding 0.5 and graph sizes greater than 160 nodes, the average degree of the Ising models surpasses 80 interactions per variable, which is more than five times the degree of the qubits in the hardware. Consequently, the algorithm addresses this by constructing qubit chains that are so extensive that it is probable that every pair of variable chains intersects at some point within the hardware graph. As a result, beyond a certain density value, it becomes equally challenging to embed problems of any density. This phenomenon is examined in greater detail by analyzing the ACL of the valid embeddings, as illustrated in Figure \ref{fig:imagenes}.
    
    \item The last observation pertains to the comparison between Minorminer's boundary and that of CE, represented by the red line. For sizes approaching 175 and densities greater than 0.5, Minorminer does not consistently find valid embeddings, unlike CE. This outcome is expected for complete graphs or density values near 1, where CE is typically recommended due to its deterministic nature. However, Minorminer remains inconsistent for lower density values, making CE a more reliable candidate for densities down to 0.5. When the density falls below this threshold, Minorminer begins to outperform CE in finding valid embeddings, with its performance improving as the instances become sparser, as anticipated.

\end{itemize}

\paragraph{Embedding quality}
As have been seen in RQ1, the quality in ACL of embeddings has a critical influence in the performance of quantum annealing. Also, the ability to embed larger and denser instances is heavily dependent on constructing embeddings with the lowest possible ACL values. In Figure \ref{fig:imagenes}, we examine the ACL of Minorminer-constructed embeddings across various ER graphs of different sizes and densities. This analysis aims to explore the dependence of ACL on these two variables, thereby enhancing our understanding of the embeddableness boundary curve. The following are the main observations from this figure:

\begin{figure*}[h!]
    \centering
    \includegraphics[scale = 0.4]{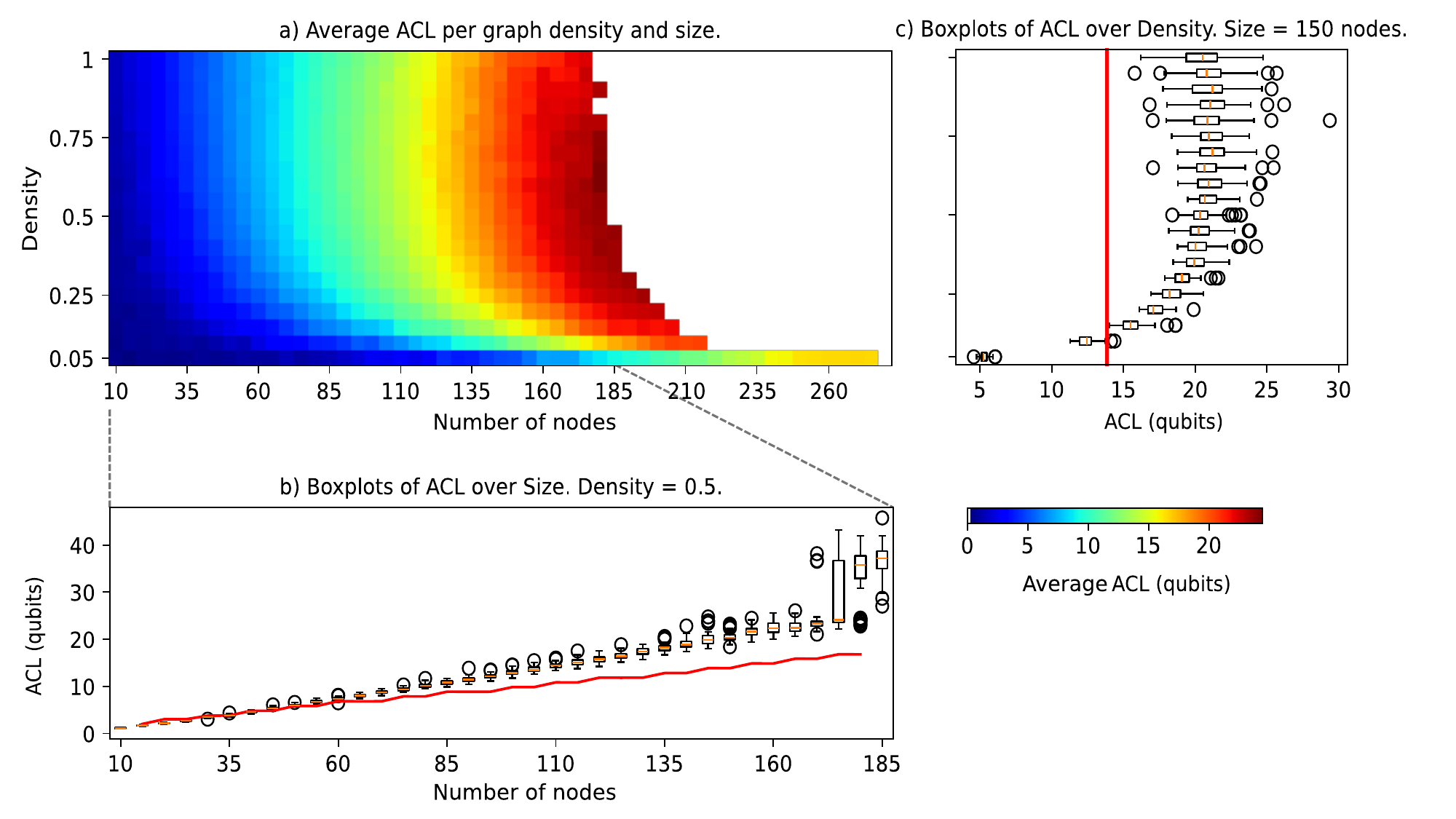}
    \caption{(a) Average ACL of valid embeddings from 64 ER graphs, per size and density. Invalid embeddings were excluded from the average.
    (b) Boxplots of ACL values of the 64 embeddings of the given size, for graph density fixed to 0.5. Red line represents the ACL of the CE embeddings.
    (c) Boxplots of ACL values of the 64 embeddings of the given density, for graph size fixed to 150. Red line represents the ACL of the CE embeddings. Boxplots display the median (orange line), the interquartile range (IQR, box), whiskers extending to 1.5×IQR, and outliers beyond as fliers.}
    \label{fig:imagenes}
\end{figure*}

\begin{itemize}
    \item Panel (a) in Figure \ref{fig:imagenes} illustrates the average ACL derived from 64 embeddings for each problem size and density values, considering only valid embeddings. The following trend can be observed in this figure: for densities below 0.5 the ACL rapidly increases with the density, but for densities greater than this value, the average ACL does not appear to vary with density and correlates solely with the size of the instances. Conversely, for densities below this value, there is a curve similar to that observed at the embedding boundary. In the sparsest instances, the ACL decreases rapidly as the density decreases. This phenomena is deeper studied in the observations of panel (c).

    \item Panel (b) features problem graphs with a fixed density of 0.5. Instead of displaying the average ACL, this panel presents boxplots that show the complete data distribution. This density value has been chosen to avoid bias towards either sparse or highly connected graphs. The boxplots in the figure exhibit a linear growth with problem size, similar to the linear growth presented by Clique Embedding, represented by the red line. Notably, when the size of the instances surpasses 60 nodes, the red line falls below the boxplots, indicating that CE has returned embeddings using a smaller number of qubits. Given that these 0.5 density instances are far from being completely connected— the specific scenario in which CE is expected to outperform Minorminer in terms of qubit usage—this suggests that Minorminer performs worse than anticipated.
    
    Moreover, as the size of the Ising models approaches the embeddable limit, the heuristic nature of the algorithm leads to increased instability, which is reflected in the suddenly wide range of the last boxplots. This instability results in a higher probability of Minorminer returning suboptimal embeddings. The stability of these results is examined in greater detail later in Figure \ref{fig:er_performance_variance}.
    
    \item Panel (c) displays boxplots of the ACL for embeddings identified in problem graphs with a fixed size of 150 nodes. This size has been selected to be close to Minorminer's embeddableness boundary observed in Figure \ref{fig:er_frontier}. Three key observations can be made solely from this panel.
    
    \begin{itemize}
        \item The first observation to be made is the low density value at which the boxplots (Minorminer's embeddings) present a higher ACL than the red line (CE's), revealing a strong limitation in Minorminer's ability to minimize qubit usage in embeddings. This happens between the density values 0.10 and 0.15, once again, significantly far of being close in density to completely connected graphs. More precisely, between 0.10 and 0.15 is the region in which the instance graphs surpasses the hardware topology graph in average degree (average number of edges per node). This observation provides insight into characterizing the Ising models that present challenges so difficult for Minorminer that it falls behind CE in embedding quality. A more general comparison between the qualities of Minorminer's and CE's embeddings is illustrated in Figure \ref{fig:mm_vs_clique}, where the relationship between this comparison and the average degree of the instances is examined in greater detail.

        \item The second observation is that as the graph density increases, the median ACL values initially rise; however, the curve plateaus once the density surpasses 0.4. This phenomenon is strongly related to the flattening of the embedableness boundary shown in Figure \ref{fig:er_frontier}. As mentioned in the observations of that figure, a possible explanation for this phenomenon lies in the challenges of embedding graphs with average degrees significantly higher than the average degree of the qubits in the hardware. Specifically, instances with 150 nodes and a density of 0.4 present an average degree of 60, which is four times that of Pegasus. Minorminer struggles with such embeddings, resulting in a network of extensive chains of around 20 qubits. Given that Pegasus has a diameter of 33—the maximum distance between any pair of nodes—it is likely that every pair of variable chains intersects at some point within the hardware graph. Consequently, with a similar number of qubits, it becomes feasible to embed even a fully connected graph by simply activating every coupler at those chain intersections. Note that this explanation should be further investigated in future work; however, this task is challenging because the minor-embedding problem is computationally intractable for these instances, and therefore Minorminer's results cannot be compared with the optimal embeddings.

        \item The third observation has to do with the range of the boxplots as the density increases. The range expands rapidly, indicating a high dispersion within the recorded ACL values for each density. For instance, in the 150-node instances, at a density of 0.7, the ACL values span from approximately 19 to 24, resulting on embedded Ising models of sizes spanning from 2850 to 3600 qubits. This substantial variation can lead to significantly poorer QA performance, as discussed in RQ1 (Section \ref{sec:embedding_quality_affection}). Additionally, the increasing range of the boxplots with density, combined with the flattening of their median values, suggests that there is potentially a higher likelihood of obtaining a better ACL valued embedding when calculating embeddings for fully connected instances compared to those with a density of 0.5. However, this preliminary observation warrants further in-depth study.
    \end{itemize}
\end{itemize}

As a part of the study of the quality of Minorminer's embeddings, motivated by the unexpectedly wide range of Ising models in which CE happens to outperform it, a broader comparison between embeddings ACLs of the two methods has been conducted. This comparison is presented in Figure \ref{fig:mm_vs_clique}, which illustrates the mean difference in ACL between the two methods for each size and density value.

\begin{figure}[t]
    \centering
    \includegraphics[width=0.95\linewidth]{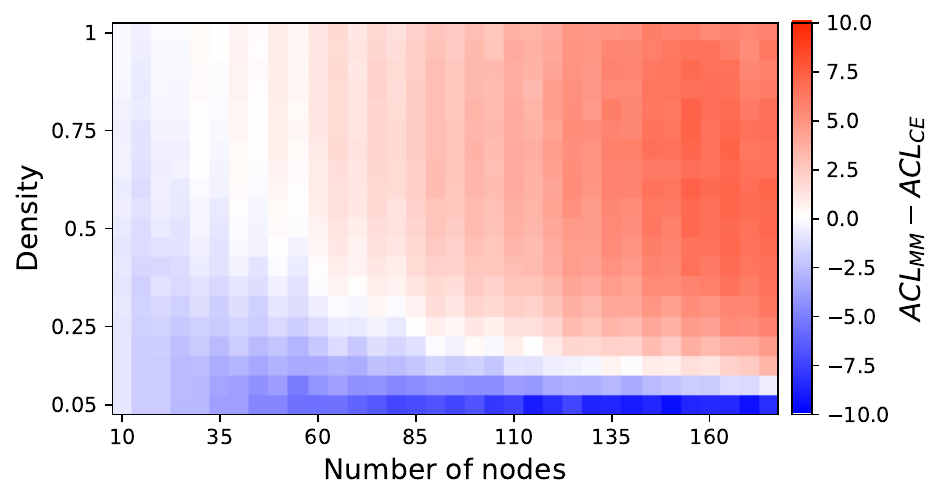}
    \caption{Difference in ACL between Minorminer's and CE's embeddings. The blue region corresponds to instances where Minorminer achieves better ACLs, while the red region indicates where CE performs better.}
    \label{fig:mm_vs_clique}
\end{figure}

It is anticipated that CE would excel in problems with high connectivity, as indicated by the flattening of the curve in Figure \ref{fig:imagenes}c, which suggests excessive qubit usage for non-sparse instance embeddings. However, as the problem size increases, CE outperforms Minorminer in instances with a density lower than 0.25. Upon examining the region where CE yields superior embeddings, it has been observed that this occurs precisely in instances where the average degree of the source graph surpasses that of the topology graph. This observation is intuitively reasonable due to the high complexity of the minor-embedding problem, particularly in cases where the source graph is more connected per node than the target graph. In such scenarios, the problem becomes critically challenging for a greedy search algorithm like Minorminer. In this experiment, the topology graph used has been a broken Pegasus. However, the same phenomenon has also been observed in a preliminary experiment with the new Zephyr topology of D-Wave \texttt{Advantage2\_prototype}. These preliminary results warrants further in-depth study.

\paragraph{Performance dispersion} Regarding the stability of the algorithm's embeddings quality, Figure \ref{fig:er_performance_variance} illustrates the standard deviation of ACLs for the 64 calculated embeddings per size and density, thereby showing the algorithm’s performance stability. 

\begin{figure}[t]
   \centering
   \includegraphics[height = 3.2cm]{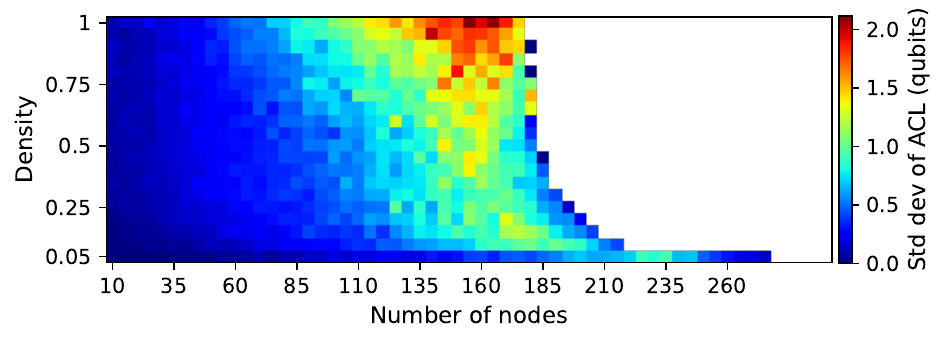}
   \caption{Standard deviation of ACLs calculated from the valid embeddings obtained during the 64 executions of Minorminer per size and density.}
   \label{fig:er_performance_variance}
\end{figure}

Minorminer exhibits high stability in small instances, indicated by a standard deviation lower than 0.5 qubits; and a notable stability in low-density larger ones, with this value slightly higher at around 1 qubit per chain. However, as instance size and density increase, the standard deviation significantly grows, reaching up to 2 qubits per chain. This variability implies that Minorminer can produce embeddings with differences in ACL of up to 4 qubits per chain for the same problem. Consequently, this not only results on a reduced likelihood of achieving the optimal embedding in a single run, but also on a significant probability of obtaining a notably bad embedding. Note that in the vicinity of the embeddableness boundary, the standard deviation decreases suddenly. This phenomenon occurs because the study only considers valid embeddings. Consequently, this reduction does not imply that the algorithm is stable; rather, it reflects the limited number of valid embeddings available.

The substantial dispersion in Minorminer's embedding ACL values is particularly concerning within the quantum annealing framework of D-Wave, where the algorithm is typically executed only once to expedite the overall process. Therefore, there is a considerable likelihood that the embedded Ising model will be significantly larger than necessary, resulting in a significantly higher error rate in the end-to-end quantum annealing method, as seen in Section \ref{sec:embedding_quality_affection}.

\begin{figure}[t]
   \centering
   \includegraphics[height = 3.2cm]{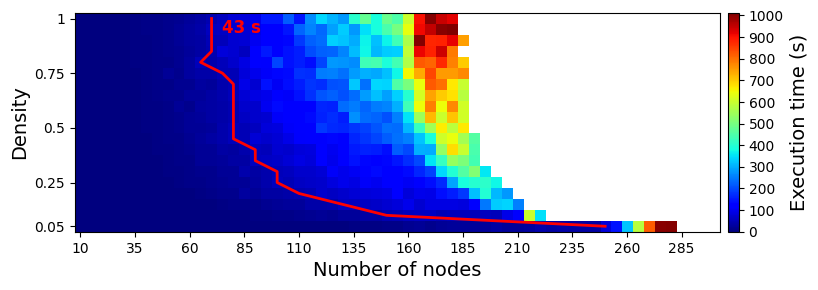}
   \caption{Execution time of Minorminer for each experiment, measured as the time required to compute 64 embeddings in parallel using the 64-core machine described in Section \ref{sec:setup}. The red line represents the total time required by CE, including both preprocessing and the maximum observed embedding retrieval time of 43 seconds.}
   \label{fig:execution_times}
\end{figure}

\paragraph{Execution time}Before comparing the runtime performance of Minorminer and Clique Embedding (CE), it is crucial to note a fundamental difference between them. CE includes a one-time preprocessing step for a given target graph, which runs in polynomial time. On our system, this step takes approximately 43 seconds for the D-Wave \texttt{Advantage\_system4.1}'s broken Pegasus target graph. Once completed, CE can embed any compatible source graph in around $10^{-5}$ seconds per instance—faster than Minorminer even for the smallest cases.

In contrast, Minorminer constructs each embedding from scratch, using both source and target graphs, and iteratively refines the result. Figure \ref{fig:execution_times} shows Minorminer’s execution times, with a red line indicating CE’s worst-case runtime under the conservative assumption that preprocessing is repeated for every instance—though in practice, it is only performed once.

Execution time increases with both problem size and graph density. Larger instances require more computation per iteration, while denser graphs lead to more complex vertex-model configurations. Additionally, runtime grows with embedding difficulty, as Minorminer uses a greedy algorithm with a patience parameter (default 10), halting only after a fixed number of iterations without improvement. As a result, harder problems take longer to embed.

The most notable observation is the absolute runtime: for the largest and densest graphs, Minorminer frequently hits the 1000-second limit while searching for valid embeddings. Only for small and sparse cases is it faster than CE’s one-time preprocessing.

As a closing remark, it is worth noting that, beyond the standard use of Minorminer, some alternative strategies have gained attention. In particular, the concept of leveraging graph layout information to generate initial embeddings—subsequently refined by Minorminer\footnote{\url{https://docs.dwavequantum.com/en/latest/ocean/api_ref_minorminer/source/layout_embedding.html}}—was independently introduced in \cite{pinilla2019layout} and \cite{zbinden2020embedding}, and has since been incorporated into D-Wave’s Ocean SDK under the name Layout Embedding. While the main objective of this work has been to conduct a thorough and rigorous evaluation of the most widely used embedding methods, a comparative analysis including Layout Embedding is provided in ~\ref{app:layout}, offering additional context on its performance relative to both Minorminer and CE.

\begin{Summary}[title={\bfseries Summary RQ2}]{}{secondsummary}
    \textcolor{black}{Across all experiments embedding Erdős–Rényi instances of varying sizes and densities into Pegasus, Minorminer—D‑Wave’s default minor-embedding algorithm—frequently underperforms relative to Clique Embedding, even though CE is designed for worst‑case qubit usage on non‑complete graphs.}
Key experimental findings include:
\begin{itemize}
    \item Embedding Success: Minorminer fails to outperform CE in finding valid embeddings for graph densities between 0.5 and 1.
    \item Embedding Quality: Minorminer only achieves shorter average chain lengths when the source graph’s average degree is lower than that of the hardware graph.
    \item Consistency: Minorminer exhibits high variability in quality for larger, denser problems, while CE remains deterministic.
    \item Speed: When including preprocessing time, CE is faster in over half the cases, often by orders of magnitude, particularly for the largest and densest instances.
\end{itemize}
Conclusion: For Ising problems that fall within CE’s embeddable range, CE is preferable—not only for fully connected graphs but also for moderately connected ones—due to its higher quality, greater consistency, and polynomial runtime performance.

\end{Summary}

\section{Final Remarks and Future Directions}\label{sec:conclusions}

Despite significant advancements in the connectivity of quantum annealing processors, the realization of fully connected topology hardware on superconducting-based quantum annealers remains impractical due to their inherently two-dimensional architecture. Furthermore, optimizing quantum annealing performance necessitates addressing the minor-embedding problem with maximal efficiency, particularly concerning the required amount of qubits. Consequently, the minor-embedding problem is likely to persist as a substantial challenge in the future.

As a result of the conducted experiments, our investigation has identified the following answers to the research questions posed:
\begin{itemize}
    \item RQ1: The minor-embedding problem has a critical impact on the performance of the end-to-end quantum annealing paradigm. \textcolor{black}{This is supported by the observation that, when comparing embedded instances of different sizes but with the same number of qubits in the embedding, the smaller problems yield worse solutions than the larger ones. It is further reinforced by the} clear correlation between the average chain length of the embeddings and the relative errors of the solutions obtained—even under optimal parameter settings.
    \item RQ2: The default minor-embedding algorithm, Minorminer, \textcolor{black}{has shown significant room for improvement when used to embed Erdös-Rényi graphs into a broken Pegasus}. CE outperforms Minorminer in terms of embedding quality, result stability, and execution time across a significantly large set of non-sparse and non-small instances, where CE should be taken as the worse-case scenario. This set is larger than just the fully connected or nearly fully connected instances, which are theoretically the problems for which CE was created and is recommended to be used. Specifically, CE has outperformed Minorminer in embedding quality for all Ising models of size embeddable by CE (up to 177 nodes) with an average degree higher than that of the topology graph.
\end{itemize}


For future research, two primary directions are proposed, as illustrated in Figure \ref{fig:futurework}. 
The first direction involves extending the experimental framework to strengthen the conclusions of this study. \textcolor{black}{In relation to the first research question, it would be valuable to examine how different chain‑resolution strategies—beyond majority voting—affect the outcomes of experiments with embedded instances. Concerning the second research question, }
the relationship between the ACL differences of embeddings produced by Minorminer and CE and the average degree of the input graph—discussed in Figure \ref{fig:mm_vs_clique}—should be investigated more thoroughly across a larger set of target graphs. It would also be valuable for future research to conduct a deeper investigation into the impact of chain length uniformity on quantum annealing performance. Previous studies have highlighted this property as beneficial in minor embeddings \cite{boothby2016fast, venturelli2015quantum}. 

To strengthen the conclusions of this study, \textcolor{black}{although the Erdös-Rényi algorithm generates a wide range of varied graphs, the study could be expanded to include new graphs generated with different algorithms, such as d-regular graphs, k-nearest neighbor graphs in Euclidean point sets \cite{penrose2003random} or Barabási-Albert graphs \cite{barabasi1999emergence} as in \cite{zbinden2020embedding}}. Additionally, this study could be repeated for any new quantum annealer relevant to the literature, whether it is a new model from D-Wave, such as the \texttt{Advantage2}, or from a different manufacturer. Furthermore, the experiments should be extended to test the performance of other algorithms.

\begin{figure}[t]

\begin{adjustbox}{minipage=0.45\textwidth,valign=t}
\begin{Futurework}[title={\bfseries Extension of Experimental Scope:}]{}{futurework1}
\begin{itemize}
    \item Explore a broader range of benchmark problem graphs
    \item Evaluate similar results across different quantum annealing hardware platforms
    \item Compare additional embedding algorithms \textcolor{black}{and broken chain resolution methods}
    \item Analyze the impact of embedding chain length uniformity on solution quality and performance
    \item Examine the relation between Minorminer and Clique Embedding performance for other problem and topology graphs
\end{itemize}
\end{Futurework}
\end{adjustbox}
\hfill
\begin{adjustbox}{minipage=0.45\textwidth,valign=t}
\begin{Futurework2}[title={\bfseries Innovative Embedding Strategies:}]{}{futurework2}
\begin{itemize}
    \item Utilizing a combination of CE and Minorminer to embed problems
    \item Enhancing solution quality through multiple executions (potentially parallel) of Minorminer
    \item Modifying Minorminer to a problem-aware algorithm
\end{itemize}
\end{Futurework2}
\end{adjustbox}
\caption{Two directions of future work.}
\label{fig:futurework}
\end{figure}

The second direction of future research involves exploring novel embedding strategies, including modifications based on Minorminer, as discussed in \cite{pinilla2019layout} and \cite{zbinden2020embedding}, as well as entirely different approaches to the problem, such as those presented in \cite{Sugie, bernal2020integer, goodrich2018optimizing, ngo_charme_2024}. Potential straightforward ideas derived from this investigation for developing new methods include:
\begin{itemize}
    \item \textit{Utilizing a combination of CE and Minorminer to embed problems}: An example of this approach can be seen in \cite{zbinden2020embedding}, where the authors propose initializing an embedding with CE and then refining it with Minorminer as a local search improvement method.
    
    \item \textit{Enhancing solution quality through multiple executions of Minorminer}: Based on the variability of Minorminer's results illustrated in Figure \ref{fig:er_performance_variance}, there are certain problems where the algorithm exhibits high instability. In these cases, executing Minorminer multiple times can significantly enhance the final solution's quality. Although multiple executions require more time, they could be performed in parallel, if feasible.

    \item \textit{Modifying Minorminer to a problem-aware algorithm}: This process involves considering the weight of variables and interactions for minor-embedding. For example, it may include embedding variables with a higher impact on solution energy into shorter chains, with the objective of prioritizing the accurate assessment of these high-impact variables during problem-solving. Even if this idea does not necessarily lead to better results on ACL of embeddings, it could lead to a lower impact of ACL on the final performance of the annealer.
\end{itemize}

Beyond the strategies mentioned above, this work justifies both previous and future efforts to develop better minor-embedding solutions. Moreover, we encourage other teams to consider this work not only when designing new minor-embedding algorithms but also when developing problem partitioning techniques, new hybrid algorithms, or related methodologies.

\section*{Acknowledgements}


This work was supported by the Basque Government through the BIKAINTEK PhD support program as well as from the HORIZON-CL4-2022-QUANTUM-01-SGA project 101113946 OpenSuperQPlus100 of the EU Flagship on Quantum Technologies. The authors used Microsoft Copilot to enhance the language and clarity of the paper. The authors retain complete responsibility for the content of this research.

\bibliography{biblio}

\appendix

\section{Small and Medium inscances on RQ1.2}
\label{app:small_medium}

\textcolor{black}{The experiment carried out to answer RQ1 in subsection RQ1.2 was executed on five instances randomly generated with Erdös-Rényi's method of 150 nodes and 0.5 density. However, this experiment was also carried out in other five random instances of 50 nodes (small) and in other five of 100 nodes (medium). The results on small and medium instances are presented in this appendix.}

\begin{figure*}[t]
    \centering
    \includegraphics[width = 1\linewidth]{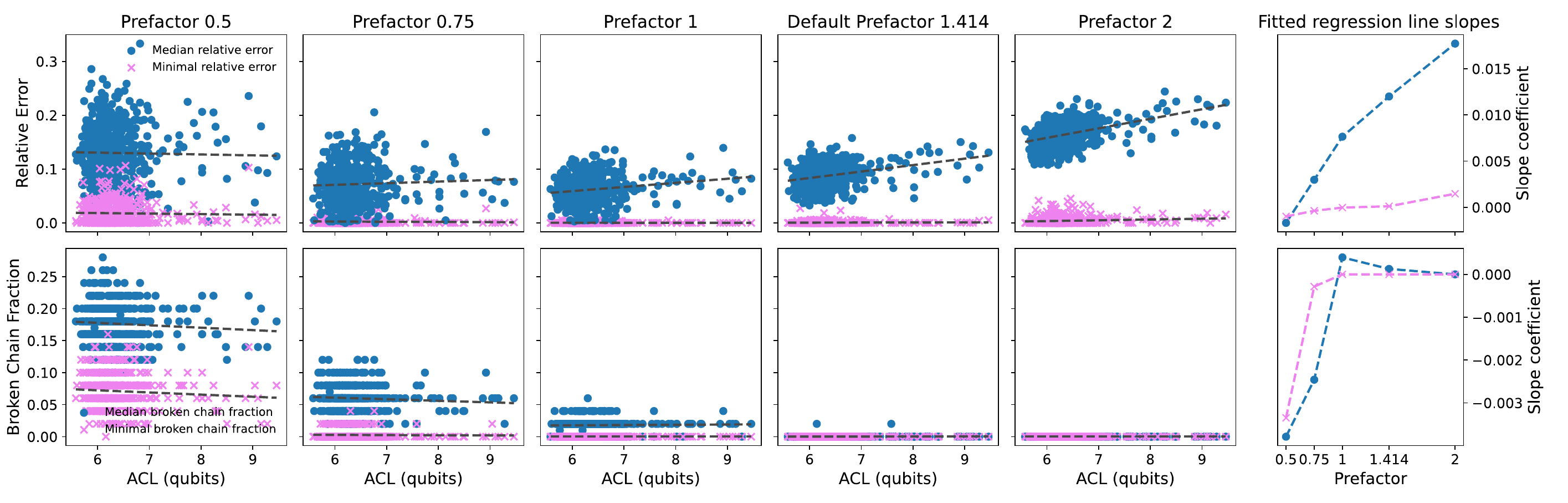}
    \includegraphics[width = 1\linewidth]{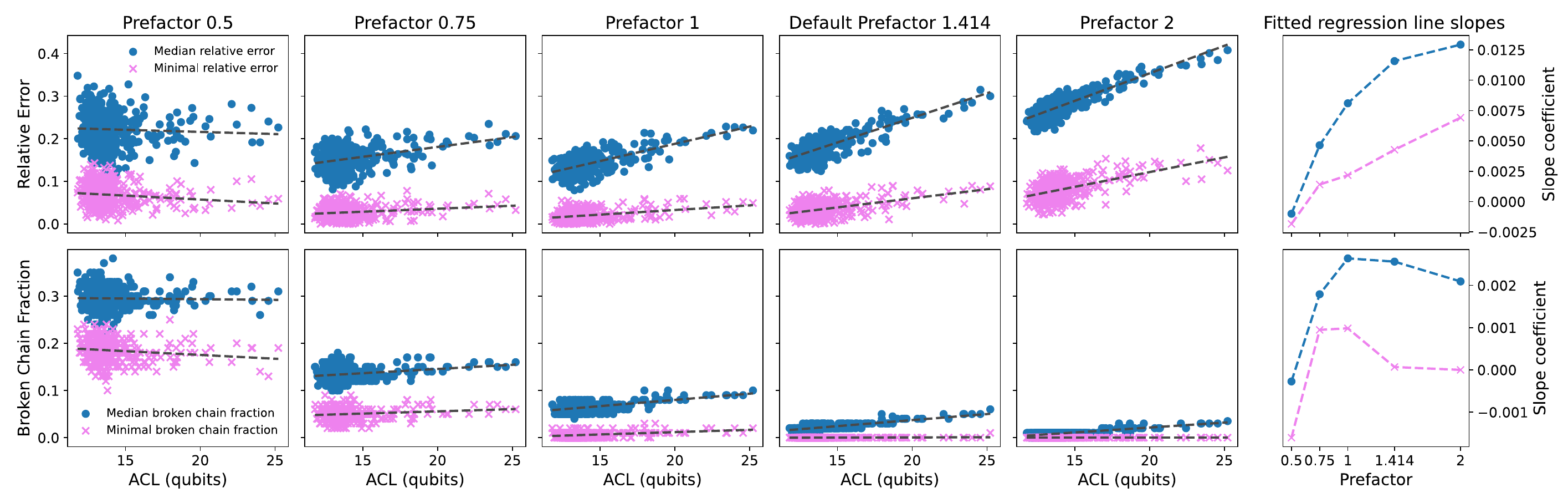}
    \caption{\textcolor{black}{Top: results for 50‑node instances; bottom: results for 100‑node instances. The first five top panels show the median and minimum relative errors of postprocessed solutions from 100 samples, each produced using a different embedding, for five randomly selected problem instances (see Table \ref{table:setup}). The first five bottom panels show the median and minimum broken‑chain fractions plotted against the ACL of the embeddings. The last column reports the slopes of the linear regressions fitted for each chain‑strength prefactor. For each problem instance and embedding, five embedded Ising models were generated by varying the chain strength via the prefactor in the UTC function, with each chain‑strength configuration represented as a separate column, and linear regression models were fitted to each corresponding set of points.}}
    \label{fig:RQ1.2_small}
\end{figure*}

\textcolor{black}{The results are, as expected, consistent with those obtained for the large instance presented in Section\ref{sec:embedding_quality_affection}, although the trends are not as pronounced. In both the small and medium cases, the optimal chain strength is typically the value obtained with a UTC prefactor of 1, though in some instances—usually those with higher ACL requirements—the optimum shifts to 0.75. As in the large instance, slightly reducing the chain strength can improve performance for embeddings with high ACL. Nevertheless, embeddings with lower ACL still produce better solutions overall, provided the chain strength is not reduced excessively. Although the conclusions for the small instance are less clear—mainly due to its lower ACL values—the overall trends remain unchanged. Specifically, the optimal prefactor is stable across problem sizes; excessively low prefactors lead to near‑random sampling, while excessively high prefactors degrade solution quality despite reducing the broken‑chain fraction. Moreover, both the relative error and the broken‑chain fraction exhibit an approximately linear increase with ACL. The behavior of the slopes mirrors that of the large instance: the slope of the error increases rapidly at first and then more slowly as chain strength grows, whereas the slope of the broken‑chain fraction initially rises due to the near‑random regime and subsequently decreases once higher chain strengths substantially suppress chain breaks.}

\section{Layout Embedding performance}
\label{app:layout}

The most prominent and practically applicable method among the cited works in Section \ref{sec:intro} of this paper is the one introduced independently by Pinilla et al. \cite{pinilla2019layout} and Zbinden et al. \cite{zbinden2020embedding}. Referred to as "\textit{Layout Aware Embedding}" by the former and "\textit{Spring-Based Minorminer}" by the latter, both approaches are fundamentally similar. Notably, this method has recently been integrated into D-Wave’s SDK under the name \textit{Layout Embedding }(LE), making it readily accessible for experimentation. In this appendix, the performance of this method and its comparison with Minorminer and Clique Embedding is shown.

The experiment conducted is identical to the one performed using Minorminer to address RQ2, as detailed in Section \ref{sec:setup}. The LE method run is the one in the D-Wave Ocean SDK\footnote{\url{https://docs.dwavequantum.com/en/latest/ocean/api_ref_minorminer/source/layout_embedding.html}}, and both its "\textit{layout}" and "\textit{placement}" functions have been set to their default value. As in Subsection \ref{sec:MMPerformance}, this appendix examines the success rate, quality, variability, and execution time of LE. As noted in prior work \cite{pinilla2019layout, zbinden2020embedding}, this method tends to outperform Minorminer on sparse instances, but underperforms on denser ones.

\paragraph{Embeddableness boundary} The goal is to identify the boundary between embeddable and non-embeddable instance graphs on D-Wave’s topology. While this boundary depends on the hardware, it is also shaped by the limitations of the heuristic embedding algorithm. Since the problem is intractable, we estimate this boundary probabilistically, based on the likelihood of finding a valid embedding. In order to see these values in perspective, we compare them with those by Minorminer and with the embedableness boundary of CE.

\begin{figure}[t]
    \centering
    \includegraphics[scale=0.4]{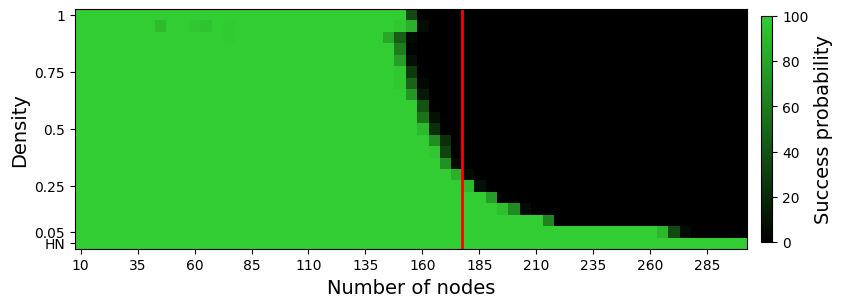} \\
    \includegraphics[scale=0.4]{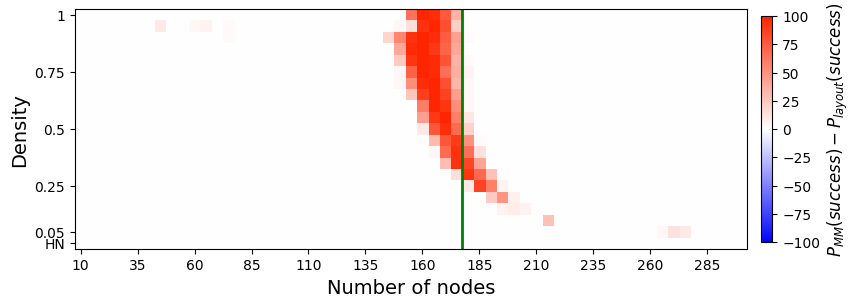}
    \caption{(Top) Probability of success of Layout Embedding finding a valid embedding for ER graphs into D-Wave \texttt{Advantage\_system4.1}'s broken Pegasus, as a function of density and size of problem graph. (Bottom) Difference in probabiliy between Minorminer and LE results. The red line (green on bottom image) represents the embeddable limit for CE. HN corresponds to Hardware Native instances, taken as the lowest significant density value.}
    \label{fig:layout_frontier_dual}
\end{figure}

As shown in Figure \ref{fig:layout_frontier_dual}, in dense instances, LE shows a lower probability of finding valid embeddings compared to Minorminer. For sparser instances, no clear advantage was observed, which may be due to suboptimal parameter settings or the specific characteristics of the instances used in this study.

\paragraph{Embedding quality}
As shown in RQ1, low ACL values are key to quantum annealing performance and to embedding larger, denser instances. Figure \ref{fig:layout_performance} analyzes how ACL varies with graph size and density for LE embeddings, helping to better understand the embeddability boundary. Key observations include:

\begin{figure}[t]
    \centering
    \includegraphics[scale=0.4]{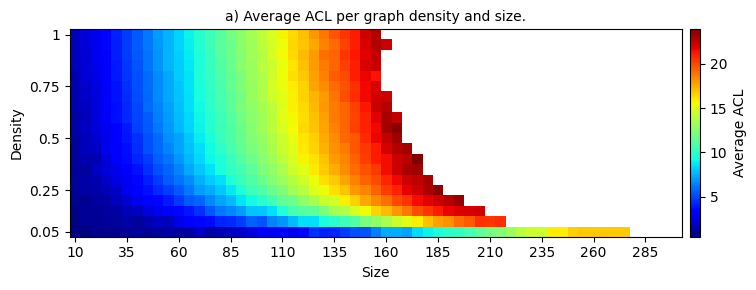} \\
    \includegraphics[scale=0.4]{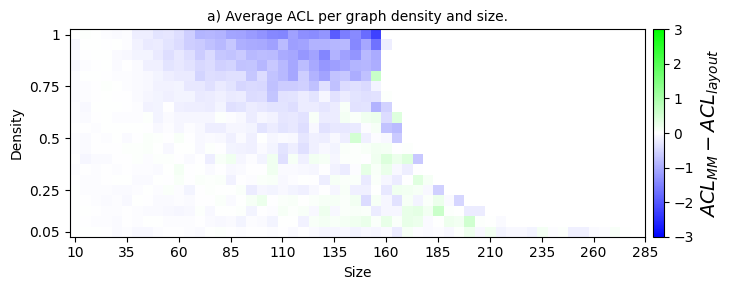} \\
    \includegraphics[scale=0.4]{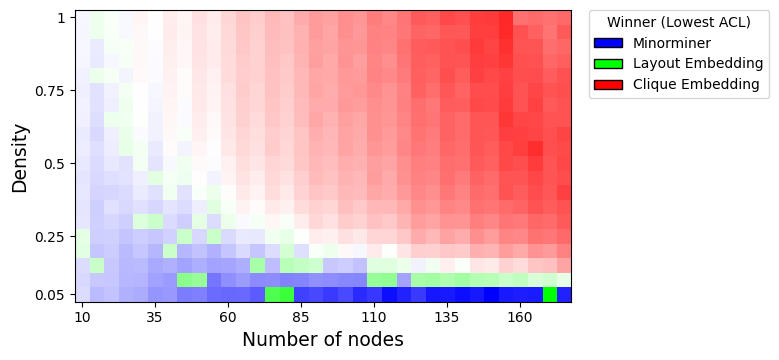}
    \caption{(Top) Average ACL of valid embeddings from 64 ER graphs, per size and density. Invalid embeddings were excluded from the average. (Center) Difference in ACL between Minorminer and LE average ACLs. Blue points indicate where Minorminer achieves the lowest ACL, while green points correspond to LE. (Bottom) ACL differences among embeddings produced by Minorminer, LE, and CE. Blue regions indicate where Minorminer achieves the lowest ACL, green where LE performs best, and red where CE yields the best results.}
    \label{fig:layout_performance}
\end{figure}

\begin{itemize}
    \item The ACL growth pattern of LE embeddings is similar to that of Minorminer, which is expected since LE runs Minorminer with a different initialization.
    \item When comparing performance between Minorminer and LE in terms of ACL, LE sometimes outperforms Minorminer on sparse instances, as noted in earlier studies, but not on denser ones.
    \item As shown in Section \ref{sec:MMPerformance}, CE is more robust than Minorminer across a range of moderately dense instances. This trend holds when comparing CE with LE as well. LE only competes with the other methods on sparse instances, and even then, inconsistently. When CE underperforms, Minorminer generally outperforms LE, though exceptions exist. However, CE dominates across a broad region where it clearly delivers superior results.
\end{itemize}

\paragraph{Performance dispersion} Figure \ref{fig:layout_performance_variance} shows the standard deviation of ACLs across 64 embeddings per size and density, highlighting the stability of the algorithm’s performance.

\begin{figure}[t]
   \centering
   \includegraphics[height = 3.2cm]{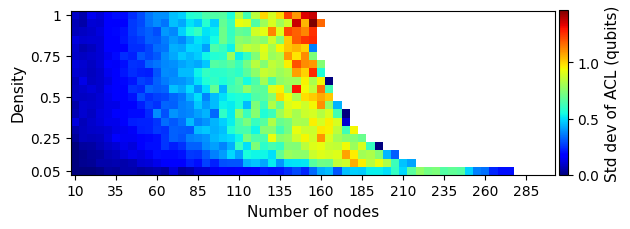}
   \caption{Standard deviation of ACLs calculated from the valid embeddings obtained during the 64 executions of Layout Embedding per size and density.}
   \label{fig:layout_performance_variance}
\end{figure}

As anticipated, LE exhibits instability growth with size and density similar to that of Minorminer, given that both rely on the same underlying algorithm with differing initializations. However, LE demonstrates slightly lower standard deviation values in ACL across embeddings (see colorbars in Figures \ref{fig:er_performance_variance} and \ref{fig:layout_performance_variance}, with maxima of 2 and 1.5, respectively). This difference is consistent with the fact that Minorminer initializes embeddings with empty chains, introducing greater randomness compared to LE’s more structured initialization.

\paragraph{Execution time} Finally, run times from Layout Embedding and their comparison with those of and CE are shown in Figure \ref{fig:layout_times}. The figures reveal that the execution time growth for Layout Embedding follows a similar trend to that of Minorminer. However, LE consistently requires significantly more time—often exceeding Minorminer by over 300 seconds—even for small and sparse instances where Minorminer completes in under 10 seconds. This discrepancy arises from the initialization phase: while Minorminer begins with empty chains, LE must first compute a layout of the source graph and assign qubit sets to variables based on distance calculations \cite{pinilla2019layout, zbinden2020embedding}. This additional processing inherently results in longer initialization times for LE.

\begin{figure}[t]
    \centering
    \includegraphics[height = 3.2cm]{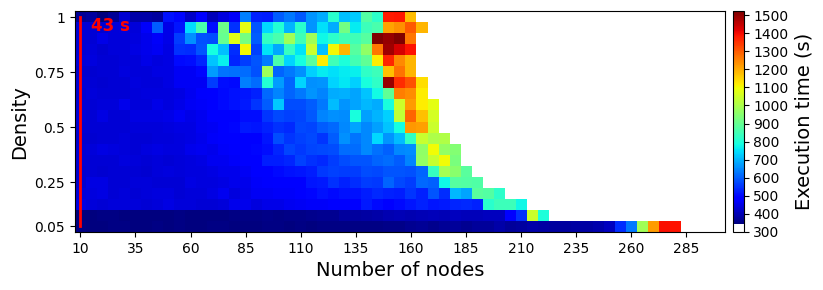} \\
    \includegraphics[height = 3.2cm]{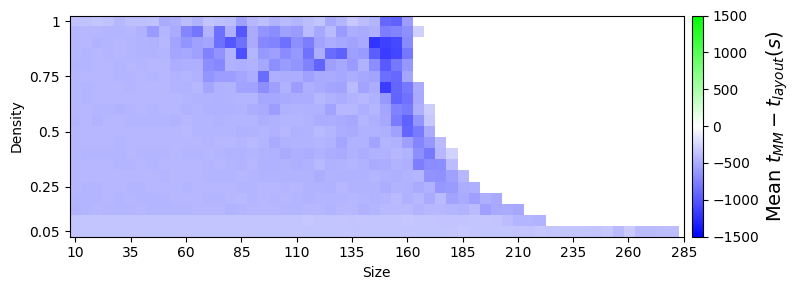}
    \caption{(Top) Execution time of Layout Embedding for each experiment, measured as the time required to compute 64 embeddings in parallel using the 64-core machine described in Section \ref{sec:setup}. The red line represents the total time required by CE, including both preprocessing and the maximum observed embedding retrieval time of 43 seconds. (Bottom) Difference in execution times between Minorminer and Layout Embedding under the same experimental conditions.}
    \label{fig:layout_times}
\end{figure}

In summary, Layout Embedding offers a promising alternative for the initialization that, in certain scenarios, outperforms direct application of Minorminer. However, it exhibits significant limitations in execution time. It is important to note that the experiments presented in this appendix are preliminary, as the evaluation of Layout Embedding was not the primary focus of this work. Further investigation, including parameter tuning, may yield different and potentially more favorable results.

The approach adopted in Layout Embedding is particularly noteworthy, as it outperforms direct use of Minorminer in certain scenarios. However, it is important to recall, as discussed in Section \ref{sec:MMPerformance}, that Minorminer itself exhibits significant limitations when compared to CE—especially in terms of solution quality and execution time across a wide range of moderately dense problems. These limitations also extend to Layout Embedding, indicating that further improvements are necessary to achieve consistently superior performance.

\end{document}